\documentclass[preprint,11pt,times,authoryear]{elsarticle}

\usepackage[section]{placeins}
\usepackage{mathdots}
\usepackage{dsfont}
\usepackage{amssymb}
\usepackage{amsthm}
\newtheorem{theorem}{Theorem}

\newtheorem{proposition}[theorem]{Proposition}

\newtheorem{remark}[theorem]{Remark}
\newtheorem{definition}{Definition}
\newtheorem{example}{Example}

\usepackage{amsmath}
\usepackage[toc,title]{appendix}
\usepackage{xspace}
\usepackage[algo2e,boxed,vlined,linesnumberedhidden,
titlenumbered]{algorithm2e}
\SetKw{KwFrom}{from}
\SetKwFor{For}{For}{do}{End for}
\SetKwFor{Forall}{For all}{do}{End for}
\SetKw{KwRet}{Return}
\SetKw{KwErr}{Error}
\SetKw{KwAnd}{and}
\SetKw{KwOr}{or}
\SetKw{KwBreak}{break}
\SetKw{KwContinue}{continue}
\SetKw{KwTrue}{true}
\SetKw{KwFalse}{false}
\SetKw{KwNext}{next}
\SetKwFor{While}{While}{do}{End do}
\SetKwIF{If}{ElseIf}{Else}{If}{then}{Else if}{Else}{End if}
\usepackage{kbordermatrix}
\usepackage{pgfplots}
\pgfplotsset{compat=newest,compat/show suggested version=false}
\pgfplotsset{
    legend image with text/.style={
        legend image code/.code={%
            \node[anchor=center] at (0.3cm,0cm) {#1};
        }
    },
}

\usepackage{ytableau}
\usepackage{url,hyperref}
\renewcommand{\kbldelim}{(}
\renewcommand{\kbrdelim}{)}

\newcommand\mcal{\mathcal}
\newcommand\mbf{\mathbf}
\newcommand\mbb{\mathbb}

\newcommand\x{\mbf{x}}

\newcommand\bin{\mbf{b}}
\newcommand\bi{\mbf{i}}

\newcommand\bk{\mbf{k}}

\newcommand\bu{\mbf{u}}

\newcommand\balpha{\boldsymbol\alpha}

\newcommand\N{\mbb{N}}

\newcommand\K{\mbb{K}}

\newcommand\F{\mbb{F}}

\newcommand\cG{\mcal{G}}
\newcommand\cK{\mcal{K}}

\newcommand\cS{\mcal{S}}
\newcommand\cT{\mcal{T}}

\newcommand\ceil[1]{\left\lceil#1\right\rceil}
\newcommand\floor[1]{\left\lfloor#1\right\rfloor}
\newcommand\cro[1]{\left[#1\right]}
\newcommand\pare[1]{\left(#1\right)}
\newcommand\acc[1]{\left\{#1\right\}}

\newcommand\BM{\textsc{BM}\xspace}
\newcommand\BMS{\textsc{BMS}\xspace}
\newcommand\aBMS{\textsc{Adaptive BMS}\xspace}
\newcommand\FGLM{\textsc{FGLM}\xspace}
\newcommand\sFGLM{\textsc{Scalar-FGLM}\xspace}
\newcommand\asFGLM{\textsc{Adaptive Scalar-FGLM}\xspace}

\newcommand\spFGLM{\textsc{Sparse-FGLM}\xspace}
\newcommand\bms{Berlekamp--Massey--Sakata\xspace}
\newcommand\bm{Berlekamp--Massey\xspace}
\newcommand\gb{Gr\"obner basis\xspace}
\newcommand\gbs{Gr\"obner bases\xspace}
\newcommand\ie{\mbox{{i.e.}}\xspace}
\newcommand\resp{\mbox{resp.}\xspace}

\DeclareMathOperator\DRL{\textsc{drl}}
\DeclareMathOperator\LEX{\textsc{lex}}

\newcommand{\adots}{\mathinner{%
  \mkern1mu\raise1pt\hbox{.}%
  \mkern2mu\raise4pt\hbox{.}%
  \mkern2mu\raise7pt\vbox{\kern7pt\hbox{.}}\mkern1mu}} 

\DeclareMathOperator\LM{\textsc{lm}}
\DeclareMathOperator\LT{\textsc{lt}}
\DeclareMathOperator\LC{\textsc{lc}}

\DeclareMathOperator\Stabilize{Stabilize}
\DeclareMathOperator\Border{Border}
\DeclareMathOperator\Staircase{Staircase}
\DeclareMathOperator\InterReduce{InterReduce}
\DeclareMathOperator\fail{fail}
\DeclareMathOperator\shift{shift}
\DeclareMathOperator\supp{supp}
\journal{Journal of Symbolic Computation}

\begin{document}

\begin{frontmatter}



\title{In-depth comparison of the Berlekamp--Massey--Sakata and
the Scalar-FGLM algorithms:\\the adaptive variants}

\author{J\'er\'emy
  Berthomieu\corref{cor1}}
\cortext[cor1]{Laboratoire d'Informatique de Paris~6,
  Sorbonne Universit\'e, Campus Pierre-et-Marie-Curie, bo\^ite
  courrier~169, 4~place Jussieu, F-75252 Paris Cedex~05, France.}
\ead{jeremy.berthomieu@lip6.fr}
\author{Jean-Charles Faug\`ere}
\ead{jean-charles.faugere@inria.fr}
\address{Sorbonne Universit\'e, \textsc{CNRS}, \textsc{INRIA},
  Laboratoire d'Informatique de Paris~6, \textsc{LIP6},
  \'Equipe \textsc{PolSys}, 4 place Jussieu, 75252 Paris Cedex 05, France}

\begin{abstract}
  The \textsc{Berlekamp--Massey--Sakata} algorithm and the
\textsc{Scalar-FGLM} algorithm both compute the ideal of
relations of a multidimensional linear recurrent sequence.

Whenever quering a single sequence element is prohibitive, the bottleneck of
these algorithms becomes the computation of all the needed sequence
terms. As such, having adaptive variants of these algorithms, reducing
the number of sequence queries, becomes mandatory.

A native adaptive variant of the \textsc{Scalar-FGLM} algorithm was
presented by its authors, the so-called \textsc{Adaptive Scalar-FGLM} algorithm.

In this paper, our first contribution is to make
the \textsc{Berlekamp--Massey--Sakata} algorithm more efficient by
making it adaptive to avoid some useless relation testings.
This variant allows us to
divide by four in dimension $2$ and by seven in dimension $3$ the number
of basic operations performed on some sequence family.

Then, we compare the two adaptive algorithms.
We show that their behaviors differ in a way that it is not possible
to tweak one of the 
algorithms in order to mimic exactly the behavior of the other. We
detail precisely the differences and the similarities of both
algorithms and conclude that in general the \textsc{Adaptive
  Scalar-FGLM} algorithm needs fewer queries and performs fewer basic
operations than the \textsc{Adaptive Berlekamp--Massey--Sakata} algorithm.

We also show that these variants are always more efficient than the
original algorithms.


\end{abstract}


\end{frontmatter}


\tableofcontents
\section{Introduction}\label{s:intro}
A fundamental problem in Computer Science is to estimate the linear
complexity of an infinite sequence $S$: this is the smallest length
of a recurrence with constant coefficients satisfied by $S$ or the
length of the shortest
linear feedback shift register (\textsc{LFSR}) which generates it.

Linear Prediction dates back to Gau\ss{} in the
18th century: given a discrete set of original values
$(u_i)_{i\in\N}$, the goal is to find the best coefficients, in the
least-squares sense, $(\alpha_i)_{i\in\N}$ that will approximate
$u_i$ by $-\sum_{k=1}^d\alpha_k\,u_{i-k}$.
Least-square sense means that the solution
minimizes the sum of the squares of the errors.

This yields a linear system whose matrix is Hankel.
This problem has also been extensively used in Digital
Signal Processing theory
and applications.  Numerically, Levinson--Durbin recursion method
can be used to solve this problem. Hence,
to some extent, the original Levinson--Durbin problem in Norbert
Wiener's Ph.D. thesis, \cite{Levinson47,Wiener49}, predates the Hankel
interpretation of the \bm algorithm, see for instance~\cite{JoMa89}.

The \bm algorithm (\BM, \cite{Berl68,Mass69}) is a famous algorithm
guessing a solution of this problem for a
one-dimensional sequence. This
algorithm has been tremendously studied and many variants were
designed.
We refer the reader
to~\cite{KaPa91,KaYa13,Kalto06} for a very nice classification
of the \BM algorithms for solving this problem, and for its
generalization to
matrix sequences.




A generalization of the \BM algorithm to $2$ dimensions was first
designed in~\cite{Sakata88}. It was then further generalized
to $n$ dimensions in~\cite{Sakata90,Sakata09}. The so-called \bms
algorithm (\BMS) guesses a \gb of the ideal of relations
satisfied by the first terms of the input sequence, \cite[Lemma~5]{Sakata90}.


In~\cite{issac2015,berthomieu:hal-01253934}, the authors designed
the \sFGLM algorithm. It also guesses a reduced
\gb of the ideal of relations of a sequence.
While the \BM algorithm can be seen as the computation of the kernel
of a Hankel matrix, the \sFGLM algorithm computes the kernel of a
\emph{multi-Hankel} matrix,
its multivariate generalization.

In some applications, computing even a term of the input sequence is
costly or even the bottleneck of the \sFGLM algorithm. An adaptive
variant of the algorithm, called the \asFGLM algorithm was designed
in~\cite{issac2015,berthomieu:hal-01253934} in order to minimize the
number of sequence queries.

More recently, the authors proposed a new algorithm,
\textsc{Polynomial Scalar-FGLM}, in~\cite{berthomieu:hal-01784369} for computing the
linear recurrence relations of a sequence based on multivariate polynomial
arithmetic. It extends the \BMS algorithm through the use of
polynomial divisions and is a complete revision of the
\sFGLM algorithm without any linear algebra operations. Yet, in this
paper the algorithms are treated as high-level ones, with linear
algebra operations. We do not try to improve them using polynomial
arithmetic as in~\cite{berthomieu:hal-01784369}.

Finally, let us recall that as it is not possible to store the
whole input sequence, 
all these algorithms take a bound as an input and only handle sequence
terms up to this index bound. This is why they can only \emph{guess}
the ideal of relations.


\subsection{Related works}
Computing linear recurrence relations of multi-dimensional sequences
finds applications in Coding Theory, Computer Algebra and
Combinatorics.

Historically, the \BM algorithm was designed to decode cyclic codes,
like the \textsc{BCH}
codes, \cite{BoseRC1960,Hocquenghem1959}. Therefore,
decoding $n$-dimensional cyclic codes, a generalization of
Reed Solomon codes, was Sakata's motivation for designing the
\BMS algorithm in~\cite{Sakata91}.

On the other hand, as the output of the \BMS
and the \sFGLM algorithms is a \gb, a
natural application in Computer Algebra
is the computation of a \gb of an ideal for another order, typically
from a total degree ordering to an elimination ordering.
In fact the latest versions of the \spFGLM algorithm
rely heavily on the \BM and \BMS algorithms,
see~\cite{FM11,faugere:hal-00807540}. These notions are recalled in a
concise way in Section~\ref{s:prelim}, see also~\cite[Section~2]{part1}.

Finally, computing linear
recurrence relations with \emph{polynomial} coefficients finds
applications in Computer Algebra for computing properties of
univariate and multivariate Special Functions. The Dynamic
Dictionary of Mathematical Functions (\textsc{DDMF}, \cite{DDMF}) generates
automatically web-pages on univariate special functions through the
differential equations they satisfy. Equivalently, they could be generated
through the linear recurrence relations satisfied by their Taylor
series sequence of
coefficients. Deciding whether
\textsc{2D}/\textsc{3D}-space walks are D-finite or not finds
applications in Combinatorics,
see~\cite{BanderierF2002,BostanBMKM2014,BousquetMM2010,BousquetMP2003}. This
motivated the authors to extend the \sFGLM algorithm to
handle relations with polynomial coefficients
in~\cite{berthomieu:hal-01314266}.

\subsection{Contributions}
Following the open question in~\cite{part1} whether an adaptive variant of
the \BMS algorithm, reducing the number of 
sequence queries, exists or not, first we answer positively. Then,  
the goal of this paper is to compare this adaptive variant and
the \asFGLM algorithm.

In Section~\ref{s:aBMS}, we design an
adaptive variant of the \BMS algorithm, namely the \aBMS algorithm,
reducing the number of sequence
queries. To our knowledge some early termination criteria were
proposed for the \BMS algorithm, see~\cite{Sakata09}. However, these
criteria did not allow one to skip some relation testings. Here, the \aBMS
algorithm can skip some relation testings and still
test some further relations. In practice, this variant is more
efficient than the \BMS algorithm thanks to these skippings. To do so,
it uses an a priori upper bound on the staircase size to
prevent some useless relation testings. In some favorable cases, this
can even allow us to require fewer sequence elements than when calling
the \BMS algorithm. The presentation of this variant follows the
linear algebra description of the \BMS
algorithm introduced in \cite[Section~3.2]{part1}, see also
Appendix~\ref{ss:bms_lin_alg}.

In Section~\ref{s:asFGLM}, we deal with the
\asFGLM algorithm, first presented in~\cite{issac2015}. Compared to
the \BMS algorithm, we iteratively increase the size of the staircase.
Although it can drastically decrease the number of sequence queries, one
of its drawback is that it can fail to compute the true ideal of
relations of a sequence.

Therefore, it is essential to investigate when these algorithms output
a \gb of the ideal of relations. To do so, we
focus on their similarities and differences of behaviors. We report
here simplified and synthetic versions of the results obtained in
Section~\ref{s:comparison_adapt}.

A first similarity is that they both output a zero-dimensional ideal
of relations.
\begin{theorem}[Theorem~\ref{th:closed_staircase_adapt}]
  Let $\bu=(u_{i,j})_{(i,j)\in\N^2}$ be a sequence, let $\prec$ be a
  degree monomial ordering and $d$ be
  the size of the staircase.
  
  Calling each algorithm on
  $\bu$, $\prec$, $d$ yields a
  truncated \gb of a zero-dimensional ideal.
\end{theorem}

In the \gb change of ordering application, like the \spFGLM algorithm,
one needs to use the lexicographical ordering. Although the \BMS
algorithm is not designed to handle such an ordering, the \aBMS can
perfectly be called with this ordering. Indeed, if the ideal is in
\emph{shape position}, then, as a second similarity, both algorithm
output correctly the ideal.
\begin{theorem}[Theorem~\ref{th:shape_position_adapt}]
  Let $\bu=(u_{i,j})_{(i,j)\in\N^2}$ be a linear recurrent sequence whose
  ideal of relations $I=\langle g(y),x-f(y)\rangle$ is in
  \emph{shape position} for the $\LEX(y\prec x)$ ordering, with $\deg
  f<\deg g=d$ and $g$ squarefree.

  Assuming no error
  is thrown in the execution of the \asFGLM algorithm called on $\bu$, $d$ and
  $\LEX(y\prec x)$ ordering,
  then the ouput ideal is $I$.

  Likewise, calling the \aBMS algorithm on $\bu$, $d$ and
  $\LEX(y\prec x)$ yields ideal $I$.
\end{theorem}

Although, the previous two theorems seem to show that both algorithms
have very similar outputs, their outputs can still differ.

Indeed, as neither algorithm can test if their output relations are valid on
the whole sequence, they intrinsically return the \emph{shifts} of the
relations: that is the set of translation monomials for which the
relations are valid. Thus, the larger the shift, the more the
relation has been tested. Therefore, it reinforces the confidence one can have
in the guessed output ideal. Even if both algorithms output the same ideal, they
usually do so while outputting different shifts.
\begin{theorem}[Theorem~\ref{th:valid_shift_adapt}]
  Let $\bu=(u_{i,j})_{(i,j)\in\N^2}$ be a sequence, $\prec$ be a
  monomial ordering and $d$ be
  the size of the output staircase $S$. Let us assume that both
  algorithms return a common relation $g$ when called on $\bu$,
  $\prec$, $d$ and some stopping monomial $M$ for the \aBMS algorithm.

  Then, the shift associated to $g$ the \aBMS algorithm yields is the
  monomial set
  $\{m,\ m\,\LM(g)\preceq M\}$. In other words, the smaller $\LM(g)$,
  the larger its shift.

  The shift associated to $g$ the \asFGLM algorithm
  returns is either $S$ if $\LM(g)\succ\max_{\prec}(S)$ or
  $\{m\in S,\ m\prec\LM(g)\}\cup\{\LM(g)\}$ otherwise. In other words,
  the larger $\LM(g)$, the larger its shift.
\end{theorem}

As a consequence of these differences of behavior, it is not possible
to tweak one of the algorithms in order to mimic exactly the behavior
of the other.

Finally, in Section~\ref{s:implem_adapt}, we compare both algorithms based on
the number of sequence queries they perform and their number of basic
operations. We show that the \aBMS algorithm is able to perform four
(\resp seven) times fewer operations than the \BMS algorithm to ouput
the ideal of relations of a family of bidimensional (\resp
tridimensional) sequences.

We also show that the \asFGLM needs fewer queries and fewer
basic operations to recover the whole ideal of relations of several
families of sequences. However, it seems that asymptotically
the ratios between the number of basic operations and the number of
sequence queries made by both algorithm could be the same.

\subsection{Conclusion and Perspectives}
We now understand better the advantages of each algorithm.

On the one hand, the \asFGLM
algorithm can fail to return the right answer, yet, on the other hand,
we can tweak it to test the computed relations further, allowing us to
discard wrong relations. Furthermore, generally it returns the right
ideal of relations and it usually does so faster than the \aBMS algorithm.

However, the \aBMS algorithm seems to be the safer one. If the upper
bounds on the staircase size is correct, it will always return the
right ideal of relations. Though, its performance speedup relies on
the number of skipped relation testings and thus on the sharpness of
this bound. Moreover, it seems hard to predict in advance which
monomials will be totally skipped during the execution of the algorithm.

Combining the design of the \textsc{Polynomial Scalar-FGLM}
algorithm, based on polynomial
arithmetic in~\cite{berthomieu:hal-01784369}, and the comparison of
the \aBMS and \asFGLM algorithms
in this paper could lead to the design of an hybrid algorithm
taking advantage of all these algorithms. In particular, this
algorithm could replace the linear algebra arithmetic by
a polynomial one.

Indeed, the goal would be to mix the efficiency of the polynomial
arithmetic in the \textsc{Polynomial Scalar-FGLM} algorithm and the
small number of queries
performed by the \aBMS and the \asFGLM algorithms to compute the relations.


\section{Preliminaries}\label{s:prelim}
In this section, we give a brief description of classical notation used all
along the paper. We refer the reader to~\cite[Section~2]{part1} for a
more detailed presentation.

\subsection{Sequences and relations}
For $n\geq 1$, we let $\bi=(i_1,\dots,i_n)\in\N^n$. Classically, we
write $\x=(x_1,\ldots,x_n)$ and
$\x^{\bi}=x_1^{i_1}\,\cdots\,x_n^{i_n}$.
An $n$-dimensional sequence $\bu=(u_\bi)_{\bi\in\N^n}$ over a field
$\K$ satisfies the
(linear recurrence) relation induced by
$\balpha=(\alpha_\bk)_{\bk\in\cK}\in\K^{|\cK|}$, with $\cK\subset\N^n$ finite if
\begin{equation}
  \label{eq:recmulti}
  \forall \bi\in\N^n,\,\sum_{\bk\in\cK}
  \alpha_\bk\, u_{\bk+\bi}=0.
\end{equation}

\begin{example}\label{ex:binom}
  Let $\bin$ be the $2$-dimensional sequence of the binomial
  coefficients, $\bin = \left(\binom{i}{j}\right)_{(i,j)\in\N^2}$.
  Then the Pascal's rule:
  \[
  \forall (i,j)\in\N^2,\, \bin_{i+1,j+1}-\bin_{i,j+1}-\bin_{i,j}=0
  \]
  is a linear recurrence relation for the sequence $\bin$.
\end{example}

As we can only work with a finite number of terms of a sequence, in
this paper, a \emph{table} shall denote a finite subset of terms of a
sequence: it is one of the input parameters of the algorithms.

Given a finite table extracted from the sequence $\bu$, the main
purpose of the \BMS and the \sFGLM algorithms is to, lousy speaking,
determine a minimal set of relations that will allow us to generate
this finite table using only the values of $\bu$ on their supports.

Relations satisfied by a sequence can be added and shifted, therefore
it is natural to associate them with multivariate polynomials in $\K[\x]$.

\begin{definition}
  Let $f=\sum_{\bk\in\cK}\alpha_\bk\,\x^\bk\in\K[\x]$.
  We will denote by $\cro{f}_{\bu}$, or $\cro{f}$ when no ambiguity
  arises, the linear combination
  $\sum_{\bk\in\cK}\alpha_\bk \,u_\bk$.
  Moreover, if $\balpha$ defines a relation for $\bu$, that is for all
  $\bi\in\N^n$, $\cro{\x^{\bi}\,f}=0$, then we say that $f$ is the
  polynomial of this relation.
\end{definition}
The main benefit of the $[\,]$ notation resides in the immediate fact
that for all index $\bi$,
$\left[\x^\bi\,f\right]=\sum_{\bk\in\cK}
\alpha_\bk\,u_{\bk+\bi}$.

In the previous example, the Pascal's
rule relation is associated with polynomial $P=x\,y-y-1$, so that
\[\forall (i,j)\in\N^2,\,[x^i\,y^j\,P]=0.\]

\begin{definition}[\cite{FitzpatrickN90,Sakata88}]~\label{def:lin_rec}
  Let $\bu=(u_{\bi})_{\bi\in\N^n}$ be an $n$-dimensional sequence
  with coefficients in $\K$. The sequence $\bu$ is \emph{linear
    recurrent} if from a nonzero finite number of initial terms
  $\{u_{\bi},\ \bi\in S\}$, and a finite number of linear recurrence
  relations, without any contradiction, one can compute any term of
  the sequence.
  
  Equivalently, $\bu$ is linear recurrent if its ideal of relations
  $\{f,\ \forall\,m\in\K[\x],\cro{m\,f}=0\}$ is \emph{zero-dimensional}.
\end{definition}

\subsection{\gbs}
Let $\cT=\{\x^{\bi},\ \bi\in\N^n\}$ be the set of all monomials
in $\K[\x]$.
A monomial ordering $\prec$ on $\K[\x]$ is an order relation
satisfying the following three classical properties:
\begin{enumerate}
\item for all $m\in\cT$, $1\preceq m$;
\item for all $m,m',s\in\cT$, $m\prec m'\Rightarrow m\,s\prec m'\,s$;
\item every subset of $\cT$ has a least element for $\prec$.
\end{enumerate}

For a monomial ordering $\prec$ on $\K[\x]$, the
\emph{leading monomial} of $f$, denoted $\LM(f)$, is the
greatest monomial in the support of $f$ for $\prec$. The
\emph{leading coefficient} of $f$, denoted $\LC(f)$, is the
nonzero coefficient of $\LM(f)$. The \emph{leading term} of $f$,
$\LT(f)$, is defined as $\LT(f)=\LC(f)\,\LM(f)$.
For an ideal $I$, we denote, classically,
$\LM(I)=\{\LM(f),\ f\in I\}$.

We recall briefly the definition of a \gb and a staircase.
\begin{definition}\label{def:staircase}
  Let $I$ be a nonzero ideal of $\K[\x]$ and let $\prec$ be a monomial ordering.
  A set $\cG\subseteq I$ is a \emph{\gb} of $I$ if for all $f\in I$,
  there exists $g\in\cG$ such that $\LM(g)|\LM(f)$.

  The set $\cG$ is a \emph{minimal} \gb of $I$ if for any $g\in\cG$,
  $\cG\setminus\{g\}$ does not span $I$.

  Furthermore, $\cG$ is (minimal) \emph{reduced} if for any $g,g'\in\cG$,
  $g\neq g'$ and any monomial $m\in\supp g'$, $\LT(g)\nmid m$.

  Let $\cG$ be a reduced truncated \gb, the \emph{staircase} of $\cG$
  is
  \[S=\Staircase(\cG)=\{s\in\cT,\ \forall\,g\in\cG, \LM(g)\nmid s\}.\]
  It is also the canonical basis of $\K[\x]/I$.
\end{definition}

\gb theory allows us to choose any monomial ordering $\prec$.
Among all the monomial ordering, we will mainly use the
\begin{itemize}
\item $\LEX(x_n\prec\cdots\prec x_1)$ ordering which compares
  monomials as follows $\x^{\bi}\prec\x^{\bi'}$ if, and only if, there
  exists $k$, $1\leq k\leq n$ such that for all $\ell<k$,
  $i_{\ell}=i_{\ell}'$ and $i_k<i_k'$,
  see~\cite[Chapter~2, Definition~3]{CoxLOS2015};
\item $\DRL(x_n\prec\cdots\prec x_1)$ order which compares monomials as
  follows $\x^{\bi}\prec\x^{\bi'}$ if, and only if,
  $i_1+\cdots+i_n<i_1'+\cdots+i_n'$ or $i_1+\cdots+i_n=i_1'+\cdots+i_n'$
  and there exists $k$,
  $2\leq k\leq n$ such that for all $\ell>k$, $i_{\ell}=i_{\ell}'$
  and $i_k>i_k'$. Equivalently, there exists $k$, $1\leq k\leq n$ such
  that for all $\ell>k$, $i_1+\cdots+i_{\ell}=i_1'+\cdots+i_{\ell}'$
  and $i_1+\cdots+i_k<i_1'+\cdots+i_k'$,
  see~\cite[Chapter~2, Definition~6]{CoxLOS2015}.
\end{itemize}

However, in the \BMS algorithm, we need to be able to enumerate all
the monomials up to a bound monomial. This forces the user to take an
ordering $\prec$
such that for all $M\in\cT$, the set $\{m\prec M,\ m\in\cT\}$
is finite. Such an ordering $\prec$ makes $(\N^n,\prec)$ isomorphic to
$(\N,<)$, thus it makes sense to speak about the next monomial
for $\prec$.

This request excludes for instance the $\LEX$ ordering, and more
generally any elimination ordering. In other words, only weighted
degree ordering, or \emph{weight ordering}, should be used.

\subsection{Multi-Hankel matrices}
A matrix $H\in\K^{m\times n}$ is \emph{Hankel}, if there exists a
sequence $\bu=(u_i)_{i\in\N}$ such that for all
$(i,i')\in\{1,\ldots,m\}\times\{1,\ldots,n\}$, the
coefficient $h_{i,i'}$ lying on the $i$th row and $i'$th column of $H$
satisfies $h_{i,i'}=u_{i+i'}$.

In a multivariate setting, we can extend this Hankel matrices notion
to \emph{multi-Hankel} matrices. Indexing the rows and columns with
monomials $\x^{\bi}=x_1^{i_1}\,\cdots\,x_n^{i_n}$ and
$\x^{\bi'}=x_1^{i'_1}\,\cdots\,x_n^{i'_n}$, the
coefficient of $H$ lying on the row labeled with
$\x^{\bi}$ and column labeled with $\x^{\bi'}$ is
$u_{\bi+\bi'}$. Given two sets of monomials $U$ and $T$, we let
$H_{U,T}$ be the multi-Hankel matrix with rows (\resp columns) indexed
with monomials in $U$ (\resp $T$).
\begin{example}
  Let $\bu=(u_{i,j})_{(i,j)\in\N^2}$ be a sequence.
  \begin{enumerate}
  \item Let $U=\{1,y,y^2,x,x\,y,x\,y^2,x^2,x^2\,y,x^2\,y^2\}$ and
    $T=\{1,y,x,x\,y,x^2,x^2\,y,x^3,x^3\,y\}$, then
    \[H_{U,T}=\kbordermatrix{
      		&1	&y	&
                &x      &x\,y	&
                &x^2	&x^2\,y	&
                &x^3	&x^3\,y\\
      1		&u_{0,0}	&u_{0,1}	&\vrule
      		&u_{1,0}	&u_{1,1}	&\vrule
                &u_{2,0}	&u_{2,1}	&\vrule
                &u_{3,0}	&u_{3,1}\\
      y    	&u_{0,1}	&u_{0,2} &\vrule
      		&u_{1,1}	&u_{1,2}	&\vrule
                &u_{2,1}	&u_{2,2}	&\vrule
                &u_{3,1}	&u_{3,2}\\
      y^2	&u_{0,2}	&u_{0,3}	&\vrule
      		&u_{1,2}	&u_{1,3}	&\vrule
                &u_{2,2}	&u_{2,3}	&\vrule
                &u_{3,2}	&u_{3,3}\\
      \cline{2-12}
      x		&u_{1,0}	&u_{1,1}	&\vrule
      		&u_{2,0}	&u_{2,1}	&\vrule
                &u_{3,0}	&u_{3,1}	&\vrule
                &u_{4,0}	&u_{4,1}\\
      x\,y	&u_{1,1}	&u_{1,2}	&\vrule
      		&u_{2,1}	&u_{2,2}	&\vrule
                &u_{3,1}	&u_{3,2}	&\vrule
                &u_{4,1}	&u_{4,2}\\
      x\,y^2	&u_{1,2}	&u_{1,3}	&\vrule
      		&u_{2,2}	&u_{2,3}	&\vrule
                &u_{3,2}	&u_{3,3}	&\vrule
                &u_{4,2}	&u_{4,3}\\
      \cline{2-12}
      x^2	&u_{2,0}	&u_{2,1}	&\vrule
      		&u_{3,0}	&u_{3,1}	&\vrule
                &u_{4,0}	&u_{4,1}	&\vrule
                &u_{5,0}	&u_{5,1}\\
      x^2\,y	&u_{2,1}	&u_{2,2}	&\vrule
      		&u_{3,1}	&u_{3,2}	&\vrule
                &u_{4,1}	&u_{4,2}	&\vrule
                &u_{5,1}	&u_{5,2}\\
      x^2\,y^2	&u_{2,2}	&u_{2,3}	&\vrule
      		&u_{3,2}	&u_{3,3}	&\vrule
                &u_{4,2}	&u_{4,3}	&\vrule
                &u_{5,2}	&u_{5,3}
    }.
    \]
    We can see that this matrix is a $3\times 4$ \emph{block-Hankel}
    matrix with Hankel blocks of size $3\times 2$.
  \item Let $T=\{1,y,x,y^2,x\,y,x^2,y^3,x\,y^2,x^2\,y,x^3\}$,
    then the following matrix has
    a less obvious structure:
    \[H_{T,T}=\kbordermatrix{
      		&1	&y	&x
      		&y^2	&x\,y	&x^2
                &y^3	&x\,y^2	&x^2\,y	&x^3\\
      1		&u_{0,0}	&u_{0,1}	&u_{1,0}
      		&u_{0,2}	&u_{1,1}	&u_{2,0}
      		&u_{0,3}	&u_{1,2}	&u_{2,1}	&u_{3,0}\\
      y		&u_{0,1}	&u_{0,2}	&u_{1,1}
      		&u_{0,3}	&u_{1,2}	&u_{2,1}
      		&u_{0,4}	&u_{1,3}	&u_{2,2}	&u_{3,1}\\
      x		&u_{1,0}	&u_{1,1}	&u_{2,0}
      		&u_{1,2}	&u_{2,1}	&u_{3,0}
                &u_{1,3}	&u_{2,2}	&u_{3,1}	&u_{4,0}\\
      y^2	&u_{0,2}	&u_{0,3}	&u_{1,2}
      		&u_{0,4}	&u_{1,3}	&u_{2,2}
      		&u_{0,5}	&u_{1,4}	&u_{2,3}	&u_{3,2}\\
      x\,y	&u_{1,1}	&u_{1,2}	&u_{2,1}
      		&u_{1,3}	&u_{2,2}	&u_{3,1}
      		&u_{1,4}	&u_{2,3}	&u_{3,2}	&u_{4,1}\\
      x^2	&u_{2,0}	&u_{2,1}	&u_{3,0}
      		&u_{2,2}	&u_{3,1}	&u_{4,0}
      		&u_{2,3}	&u_{3,2}	&u_{4,1}	&u_{5,0}\\
      y^3	&u_{0,3}	&u_{0,4}	&u_{1,3}
      		&u_{0,5}	&u_{1,4}	&u_{2,3}
      		&u_{0,6}	&u_{1,5}	&u_{2,4}	&u_{3,3}\\
      x\,y^2	&u_{1,2}	&u_{1,3}	&u_{2,2}
      		&u_{1,4}	&u_{2,3}	&u_{3,2}
      		&u_{1,5}	&u_{2,4}	&u_{3,3}	&u_{4,2}\\
      x^2\,y	&u_{2,1}	&u_{2,2}	&u_{3,1}
      		&u_{2,3}	&u_{3,2}	&u_{4,1}
      		&u_{2,4}	&u_{3,3}	&u_{4,2}	&u_{5,1}\\
      x^3	&u_{3,0}	&u_{3,1}	&u_{4,0}
      		&u_{3,2}	&u_{4,1}	&u_{5,0}
      		&u_{3,3}	&u_{4,2}	&u_{5,1}	&u_{6,0}
    }.
    \]
  \end{enumerate}
\end{example}


\section{An Adaptive version of the \BMS algorithm}\label{s:aBMS}
The \BMS algorithm was presented first in~\cite{Sakata88} for the
dimension $2$ case and then was extended to dimension $n$
in~\cite{Sakata90,Sakata09}. In~\cite[Section~3]{part1}
or Appendix~\ref{a:BMS}, we give a description of the algorithm mainly
based on linear algebra.

The \BMS algorithm is an iterative algorithm, visiting each term
$u_{\bi}=[\x^{\bi}]$ of the input sequence in increasing order for the
input monomial order. At each step, it has a truncated \gb of the
ideal of relations and test them in the visited monomial.
If some of them fail, the algorithm updates the \gb with new valid relations.

When a relation $g$ fails at monomial $m$, two situations arise:
either $\frac{m}{\LM(g)}$ was already in the staircase and then a new
relation $g'$ with $\LM(g')=\LM(g)$ is computed or it was not and both
$\LM(g)$ and $\frac{m}{\LM(g)}$ are added to the staircase. New
relations are then computed depending on the possible new leading
monomials. See~\cite[Proposition~9]{part1} and Proposition~\ref{prop:iter}.

This is summed up in the following example; it is a truncated version
of~\cite[Example~10]{part1} and Example~\ref{ex:binom_bms}.
\begin{example}\label{ex:truncated_binom_bms}
  We give the trace of the algorithm called on the binomial sequence
  $\bin$ for the
  $\DRL(y\prec x)$ ordering from monomial $y^5$ up to monomial $x^5$.

  To simplify the reading, whenever a relation succeeds in $m$ or
  cannot be tested in $m$, we
  skip the updating part as this relation remains the same.

  We start with the non empty staircase $S=\{[y^2,x^2],[x^2-2\,x+1,y^2]\}$ and
  the relations $G=\{x\,y-y-1,y^3,x^3-2\,x^2+x\}$. This means that on the one
  hand the relations in $G$ have been tested up to all their multiples
  less than $y^5$ while relation $y^2$ (\resp $x^2-2\,x+1$) in $S$ fails when
  multiplied by $x^2$ (\resp $y^2$) but does not fail when multiplied by a
  lesser monomial.
  \begin{enumerate}
  \item[] For the monomial $y^5$
    \begin{enumerate}
    \item[] Nothing must be done for the relation $g_1=x\,y-y-1$.
    \item[] The relation $g_2=y^3$ succeeds since $[\bin_{0,5}]=0$.
    \item[] Nothing must be done for the relation $g_3=x^3-2\,x^2+x$.
    \end{enumerate}
  \item[] For the monomial $x\,y^4$
    \begin{enumerate}
    \item[] The relation $g_1=x\,y-y-1$ succeeds since $[\bin_{1,4}-\bin_{0,4}-\bin_{0,3}]=0$.
    \item[] The relation $g_2=y^3$ succeeds since $[\bin_{1,4}]=0$.
    \item[] Nothing must be done for the relation $g_3=x^3-2\,x^2+x$.
    \end{enumerate}
  \item[] For the monomial $x^2\,y^3$
    \begin{enumerate}
    \item[] The relation $g_1=x\,y-y-1$ succeeds since $[\bin_{1,4}-\bin_{0,4}-\bin_{0,3}]=0$.
    \item[] The relation $g_2=y^3$ succeeds since $[\bin_{2,3}]=0$.
    \item[] Nothing must be done for the relation $g_3=x^3-2\,x^2+x$.
    \end{enumerate}
  \item[] For the monomial $x^3\,y^2$
    \begin{enumerate}
    \item[] The relation $g_1=x\,y-y-1$ succeeds since $[\bin_{3,2}-\bin_{2,2}-\bin_{2,1}]=0$.
    \item[] Nothing must be done for the relation $g_2=y^3$.
    \item[] The relation $g_3=x^3-2\,x^2+x$ fails since
      $[\bin_{3,2}-2\,\bin_{2,2}+\bin_{1,2}]=1$. Thus
      $S'=\{[y^2,x^2],[x^2-2\,x+1,y^2],[x^3-2\,x+1,y^2]\}$.
    \item[] $S'$ is set to $\{[y^2,x^2],[x^2-2\,x+1,y^2]\}$ and $G'=\{y^3,x\,y,x^3\}$.
    \item[] We set $g_1'=x\,y-y-1$ and $g_2'=y^3$.
    \item[] For the relation $g_3'=x^3$, $x^3| x^3\,y^2$ and
      $\frac{x^3\,y^2}{x^3}|\fail(x^2-2\,x+1)$, hence $g_3'=x^3-3\,x^2+3\,x-1$.
    \item[] We update $G:=G'=\{y^3,x\,y-y-1,x^3-3\,x^2+3\,x-1\}$ and
      $S:=S'=\{[y^2,x^2],[x^2-2\,x+1,y^2]\}$.
    \end{enumerate}
  \item[] For the monomial $x^4\,y$
    \begin{enumerate}
    \item[] The relation $g_1=x\,y-y-1$ succeeds since $[\bin_{4,1}-\bin_{3,1}-\bin_{3,0}]=0$.
    \item[] Nothing must be done for the relation $g_2=y^3$.
    \item[] The relation $g_3=x^3-3\,x^2+3\,x-1$ succeeds since
      $[\bin_{4,1}-3\,\bin_{3,1}+3\,\bin_{2,1}-\bin_{1,1}]=0$.
    \end{enumerate}
  \item[] For the monomial $x^5$
    \begin{enumerate}
    \item[] Nothing must be done for the relation $g_1=x\,y-y-1$.
    \item[] Nothing must be done for the relation $g_2=y^3$.
    \item[] The relation $g_3=x^3-3\,x^2+3\,x-1$ succeeds since
      $[\bin_{5,0}-3\,\bin_{4,0}+3\,\bin_{3,0}-\bin_{2,0}]=0$.
    \end{enumerate}
  \item[] The algorithm returns relations
    $x\,y-y-1,y^3,x^3-3\,x^2+3\,x-1$, the first one with a shift $x^3$
    and the last two with a shift $x^2$.
  \end{enumerate}
\end{example}

The problem is now to understand when the \gb of the ideal of
relations has actually been computed.
Assuming the sequence is linear recurrent,
Proposition~\ref{prop:upperbound_repeat} provides an answer to this question (see
also~\cite[Proposition~11]{part1} and Proposition~\ref{prop:upperbound}.
\begin{proposition}\label{prop:upperbound_repeat}
  Let $\bu$ be a linear recurrent sequence and $I$ be its ideal of
  relations.

  Let $S$ be the
  staircase of $I$ for $\prec$. Let $s_{\max}$
  be the largest monomial in
  $S$. Then, at step $m\succeq (s_{\max})^2$, the computed staircase
  is equal to $S$.

  Let $\cG$ be a minimal \gb of $I$ for $\prec$ and let $g_{\max}$ be the
  largest leading monomial of  $\cG$. Then, at step $m\succeq
  s_{\max}\cdot\max_{\prec}(g_{\max},s_{\max})$,
  the computed \gb is a minimal \gb of $I$ for $\prec$.
\end{proposition}
\begin{example}\label{ex:bound_for_rect}
  For the $\DRL(y\prec x)$ ordering, $I=\langle x^p,y^q\rangle$ and
  $q>p\geq 1$, we have, $s_{\max}=x^{p-1}\,y^{q-1}$ and
  $g_{\max}=y^q$. Therefore, the right staircase is found at most at step
  $m=x^{2\,p-2}\,y^{2\,q-2}$, while the \gb is found at most at step
  $x^{p-1}\,y^{q-1}\,\max_{\prec}(x^{p-1}\,y^{q-1},y^q)$, \ie
  $y^{2\,q-1}$ if $p=1$ and $x^{2\,p-2}\,y^{2\,q-2}$ otherwise.
\end{example}

\begin{remark}
  In some favourable cases though, it is not necessary to go up to
  this bound to guess the right relations. In
  Example~\ref{ex:bound_for_rect}, for $p=1$ and $q=2$, the right
  staircase is found at step $y$. In fact, the
  right \gb is already guessed as well, while
  Proposition~\ref{prop:upperbound_repeat} only ensures that it will
  be correctly guessed at step $y^3$.
\end{remark}
It could therefore be very
fruitful to have an heuristic helping us determining if the current \gb
is the right one when the size of the staircase is known in
advance. Indeed, it could allow us to end earlier the running of the
\BMS algorithm. Unfortunately, it is not rare that an interrupted
\BMS algorithm does not return the correct \gb, in fact such an
interrupted \BMS algorithm will never return the right \gb for any of the four
families of sequences used in Section~\ref{s:implem_adapt}. The goal
is thus to reduce the number of testings differently.

Let us recall that at step $m$, whenever a relation $g$
such that $\LM(g)|m$ fails, if $\frac{m}{\LM(g)}$ is not in the
staircase, then the algorithm adds both $\LM(g)$ and $\frac{m}{\LM(g)}$ in the
new staircase. Assuming we know in advance the size of the staircase of the output
\gb, during the execution of the algorithm, we can detect that testing
the relation $g$ in $m$ is useless if the staircase becomes too big
after adding the two monomials.

Let us show in the following example how we can take advantage of this strategy.

\begin{example}
  Let us reconsider Example~\ref{ex:truncated_binom_bms} with
  the assumption that the staircase has a size at most $5$.

  We start with the non empty staircase $S=\{[y^2,x^2],[x^2-2\,x+1,y^2]\}$ and
  the relations $G=\{x\,y-y-1,y^3,x^3-2\,x^2+x\}$. This means that on the one
  hand the relations in $G$ have been tested up to all their multiples
  less than $y^5$ while relation $y^2$ (\resp $x^2-2\,x+1$) in $S$ fails when
  multiplied by $x^2$ (\resp $y^2$) but does not fail when multiplied by a
  lesser monomial.
  \begin{enumerate}
  \item[] For the monomial $y^5$
    \begin{enumerate}
    \item[] Nothing must be done for the relation $g_1=x\,y-y-1$.
    \item[] The relation $g_2=y^3$ succeeds since $[\bin_{0,5}]=0$.
    \item[] Nothing must be done for the relation $g_3=x^3-2\,x^2+x$.
    \end{enumerate}
  \item[] For the monomial $x\,y^4$
    \begin{enumerate}
    \item[] Should the relation $g_1=x\,y-y-1$ fail in $x\,y^4$, we
      would have to add $x\,y$ and $y^3$ in the staircase, raising
      its size to $7$. We skip testing $g_1$.
    \item[] Should the relation $g_2=y^3$ fail in $x\,y^4$, we
      would have to add $y^3$ and $x\,y$ in the staircase, raising
      its size to $7$. We skip testing $g_2$.
    \item[] Nothing must be done for the relation $g_3=x^3-2\,x^2+x$.
    \end{enumerate}
  \item[] For the monomial $x^2\,y^3$
    \begin{enumerate}
    \item[] Should the relation $g_1=x\,y-y-1$ fail in $x^2\,y^3$, we
      would have to add $x\,y$ and $x\,y^2$ in the staircase, raising
      its size to $7$. We skip testing $g_1$.
    \item[] The relation $g_2=y^3$ succeeds since $[\bin_{2,3}]=0$.
    \item[] Nothing must be done for the relation $g_3=x^3-2\,x^2+x$.
    \end{enumerate}
  \item[] For the monomial $x^3\,y^2$
    \begin{enumerate}
    \item[] Should the relation $g_1=x\,y-y-1$ fail in $x^3\,y^2$, we
      would have to add $x\,y$ and $x^2\,y$ in the staircase, raising
      its size to $7$. We skip testing $g_1$.
    \item[] Nothing must be done for the relation $g_2=y^3$.
    \item[] The relation $g_3=x^3-2\,x^2+x$ fails since
      $[\bin_{3,2}-2\,\bin_{2,2}+\bin_{1,2}]=1$. Thus
      $S'=\{[y^2,x^2],[x^2-2\,x+1,y^2],[x^3-2\,x+1,y^2]\}$.
    \item[] $S'$ is set to $\{[y^2,x^2],[x^2-2\,x+1,y^2]\}$ and $G'=\{y^3,x\,y,x^3\}$.
    \item[] We set $g_1'=x\,y-y-1$ and $g_2'=y^3$.
    \item[] For the relation $g_3'=x^3$, $x^3| x^3\,y^2$ and
      $\frac{x^3\,y^2}{x^3}|\fail(x^2-2\,x+1)$, hence $g_3'=x^3-3\,x^2+3\,x-1$.
    \item[] We update $G:=G'=\{y^3,x\,y-y-1,x^3-3\,x^2+3\,x-1\}$ and
      $S:=S'=\{[y^2,x^2],[x^2-2\,x+1,y^2]\}$.
    \end{enumerate}
  \item[] For the monomial $x^4\,y$
    \begin{enumerate}
    \item[] Should the relation $g_1=x\,y-y-1$ fail in $x^4\,y$, we
      would have to add $x\,y$ and $x^3$ in the staircase, raising
      its size to $7$. We skip testing $g_1$.
    \item[] Nothing must be done for the relation $g_2=y^3$.
    \item[] Should the relation $g_3=x^3-3\,x^2+3\,x-1$ fail in $x^4\,y$, we
      would have to add $x^3$ and $x\,y$ in the staircase, raising
      its size to $7$. We skip testing $g_3$.
    \end{enumerate}
  \item[] For the monomial $x^5$
    \begin{enumerate}
    \item[] Nothing must be done for the relation $g_1=x\,y-y-1$.
    \item[] Nothing must be done for the relation $g_2=y^3$.
    \item[] The relation $g_3=x^3-3\,x^2+3\,x-1$ succeeds since
      $[\bin_{5,0}-3\,\bin_{4,0}+3\,\bin_{3,0}-\bin_{2,0}]=0$.
    \end{enumerate}
  \item[] The algorithm returns relations
    $x\,y-y-1,y^3,x^3-3\,x^2+3\,x-1$, the first one with a shift $x^3$
    and the other two with a shift $x^2$.
  \end{enumerate}
\end{example}

In this example, skipping some relation testings allowed us to
skip all the testings in a loop, namely loops $x\,y^4$ and
$x^4\,y$. As a byproduct, we also reduced the number of table queries.

Integrating this strategy in the \BMS algorithm yields an adaptive
variant, Algorithm~\ref{algo:abms},
reducing the number of relation testings and table queries.

\begin{algorithm2e}[htbp!]\label{algo:abms}
  \small
  \DontPrintSemicolon
  \TitleOfAlgo{\aBMS (Linear Algebra variant).}
  \KwIn{A table $\bu=(u_{\bi})_{\bi\in\N^n}$ with coefficients in
    $\K$, a monomial ordering $\prec$, a given bound $d$
    and a monomial $M$
    as the stopping condition.}
  \KwOut{A set $G$ of relations generating $I_M$.}
  $T := \{m\in\K[\x], m\preceq M\}$.\;
  $G := \{1\}$.\;
  $S := \emptyset$.\;
  \Forall{$m\in T$}{
    $S' := S$.\;
    \For{$g\in G$}{
      \If{$\LM(g)| m$}{
        \If{$\frac{m}{\LM(g)}\not\in\Stabilize(S)$ 
          \KwAnd $\#\,\Stabilize\pare{S\cup\acc{\LM(g),
              \frac{m}{\LM(g)}}}> d$}{
          \KwNext.\tcp*{skip this relation testing}
        }
        $e:=\cro{\frac{m}{\LM(g)}\,g}_{\bu}$\;
        \If{$e\neq 0$}{
          $S':=S'\cup
          \acc{\cro{\frac{g}{e},\frac{m}{\LM(g)}}}$.\;
        }
      }
    }
    $S':=
    \min_{\fail(h)\in S'}\acc{[h,\fail(h)/\LM(h)]}$.\;
    $G':= \Border (S')$.\;
    \For{$g'\in G'$}{
      Let $g\in G$ such that $\LM(g)|\LM(g')$.\; 
      \uIf{$\LM(g) \nmid m$}{
        $g':=\frac{\LM(g')}{\LM(g)}\,g$.\;
      }
      \uElseIf{$\exists\,h\in S,
        \frac{m}{\LM(g')} |\fail(h)$}{
        $g':= \frac{\LM(g')}{\LM(g)}\,g
        -\cro{\frac{m}{\LM(h)}\,h}_{\bu}\,
        \frac{\LM(g')\,\fail(h)}{m}\,h$.\;
      }
      \lElse{
        $g':=g$.
      }
    }
    $G := G'$.\;
    $S := S'$.\;
  }  
  \KwRet $G$.
\end{algorithm2e}
This version was motivated by a remark in~\cite{Sakata09} where the
author announced that in applications where an approximate size of the
staircase is known, one can stop early the execution of the \BMS
algorithm. Yet, we do not know if such a strategy is classical and
if it is exactly the one described in Algorithm~\ref{algo:abms}.

Predicting how many monomials will be completely skipped in order to
reduce the number of table queries can
be a hard task. Indeed, it is clear that if relation $g$ can be
skipped at monomial $m$, it will also be skipped at any multiple of
$m$. Yet, even if $m$ is completely skipped, a relation that cannot be
tested in $m$ might need to be
tested in $m\,x_i$ for some $i$.

Therefore, even if $m$ is completely skipped, $m\,x_i$ might be
not. We illustrate this phenomon with the following example.
\begin{example}
  Let $\bu=(u_{i,j})_{(i,j)\in\N^2}$ be the sequence defined by
  $u_{4,1}=1$ and $u_{i,j}=0$ if $(i,j)\neq(4,1)$. Running the \BMS
  algorithm on these arguments yields relations 
  $y^2,x^5$ so that the staircase has size $10$.
  We assume though that the only known
  upper bound on the staircase size is $14$.

  We give a short trace of the algorithm called on $\bu$ for the
  $\DRL(y\prec x)$ ordering up to
  monomial $x^9$.  Therefore, we
  also input $14$ as the upper bound on the size of the output staircase to the \aBMS
  algorithm.

  \begin{enumerate}
  \item[] For all the monomials from $1$ to $x^3\,y^2$
    \begin{enumerate}
    \item[] The relation $g_1=1$ succeeds.
    \end{enumerate}
  \item[] For the monomial $x^4\,y$
    \begin{enumerate}
    \item[] The relation $g_1=1$ fails since
      $[u_{4,1}]=1$. Thus
      $S'=\{[1,x^4\,y]\}$.
    \item[] $S'$ is set to $\{[1,x^4\,y]\}$ and $G'=\{y^2,x^5\}$.
    \item[] We set $g_1'=y^2$ and $g_2'=x^5$.
    \item[] For the relation $g_1'=y^2$, $y^2\nmid x^4\,y$ thus $g_1'=y^2$.
    \item[] For the relation $g_2'=x^5$, $x^5\nmid x^4\,y$ thus $g_2'=x^5$.
    \item[] We update $G:=G'=\{y^2,x^5\}$ and
      $S:=S'=\{[1,x^4\,y]\}$.
    \end{enumerate}
  \item[] For the monomial $x^5$
    \begin{enumerate}
    \item[] The relation $g_2=x^5$ succeeds.
    \end{enumerate}
  \item[] For the monomial $y^6$
    \begin{enumerate}
    \item[] The relation $g_1=y^2$ succeeds.
     \end{enumerate}
  \item[] For the monomial $x\,y^5$
    \begin{enumerate}
    \item[] The relation $g_1=y^2$ succeeds.
    \end{enumerate}
  \item[] For the monomial $x^2\,y^4$
    \begin{enumerate}
    \item[] The relation $g_1=y^2$ succeeds.
    \end{enumerate}
  \item[] For the monomial $x^3\,y^3$
    \begin{enumerate}
    \item[] The relation $g_1=y^2$ succeeds.
    \end{enumerate}
  \item[] For the monomial $x^4\,y^2$
    \begin{enumerate}
    \item[] The relation $g_1=y^2$ succeeds.
    \end{enumerate}
  \item[] For the monomial $x^5\,y$
    \begin{enumerate}
    \item[] The relation $g_2=x^5$ succeeds.
    \end{enumerate}
  \item[] For the monomial $x^6$
    \begin{enumerate}
    \item[] The relation $g_2=x^5$ succeeds.
    \end{enumerate}
  \item[] For the monomial $y^7$
    \begin{enumerate}
    \item[] The relation $g_1=y^2$ succeeds.
    \end{enumerate}
  \item[] For the monomial $x\,y^6$
    \begin{enumerate}
    \item[] Should the relation $g_1=y^2$ fail, we would have to add
      $y^2$ and $x\,y^4$ to the staircase, raising its size to $16$.
      We skip testing $g_1$.
    \end{enumerate}
  \item[] For the monomial $x^2\,y^5$
    \begin{enumerate}
    \item[] Should the relation $g_1=y^2$ fail, we would have to add
      $y^2$ and $x^2\,y^3$ to the staircase, raising its size to $16$.
      We skip testing $g_1$.
    \end{enumerate}
  \item[] For the monomial $x^3\,y^4$
    \begin{enumerate}
    \item[] The relation $g_1=y^2$ succeeds.
    \end{enumerate}
  \item[] For the monomial $x^4\,y^3$
    \begin{enumerate}
    \item[] The relation $g_1=y^2$ succeeds.
    \end{enumerate}
  \item[] For the monomial $x^5\,y^2$
    \begin{enumerate}
    \item[] The relation $g_1=y^2$ succeeds.
    \item[] The relation $g_2=x^5$ succeeds.
    \end{enumerate}
  \item[] For the monomial $x^6\,y$
    \begin{enumerate}
    \item[] The relation $g_2=x^5$ succeeds.
    \end{enumerate}
  \item[] For the monomial $x^7$
    \begin{enumerate}
    \item[] The relation $g_2=x^5$ succeeds.
    \end{enumerate}
  \item[] For the monomial $y^8$
    \begin{enumerate}
    \item[] Should the relation $g_1=y^2$ fail, we would have to add
      $y^2$ and $y^6$ to the staircase, raising its size to $16$.
      We skip testing $g_1$.
    \end{enumerate}
  \item[] For the monomial $x\,y^7$
    \begin{enumerate}
    \item[] We did not test $g_1$ in $x\,y^6$. We skip testing $g_1$.
    \end{enumerate}
  \item[] For the monomial $x^2\,y^6$
    \begin{enumerate}
    \item[] We did not test $g_1$ in $x\,y^6$ and $x^2\,y^5$. We skip testing $g_1$.
    \end{enumerate}
  \item[] For the monomial $x^3\,y^5$
    \begin{enumerate}
    \item[] We did not test $g_1$ in $x^2\,y^5$. We skip testing $g_1$.
    \end{enumerate}
  \item[] For the monomial $x^4\,y^4$
    \begin{enumerate}
    \item[] Should the relation $g_1=y^2$ fail, we would have to add
      $y^2$ and $x^4\,y^2$ to the staircase, raising its size to $15$.
      We skip testing $g_1$.
    \end{enumerate}
  \item[] For the monomial $x^5\,y^3$
    \begin{enumerate}
    \item[] The relation $g_1=y^2$ succeeds.
    \item[] The relation $g_2=x^5$ succeeds.
    \end{enumerate}
  \item[] For the monomial $x^6\,y^2$
    \begin{enumerate}
    \item[] The relation $g_1=y^2$ succeeds.
    \item[] The relation $g_2=x^5$ succeeds.
    \end{enumerate}
  \item[] For the monomial $x^7\,y$
    \begin{enumerate}
    \item[] The relation $g_2=x^5$ succeeds.
    \end{enumerate}
  \item[] For the monomial $x^8$
    \begin{enumerate}
    \item[] The relation $g_2=x^5$ succeeds.
    \end{enumerate}
  \item[] For the monomial $y^9$
    \begin{enumerate}
    \item[] We did not test $g_1$ in $y^8$. We skip testing $g_1$.
     \end{enumerate}
  \item[] For the monomial $x\,y^8$
    \begin{enumerate}
    \item[] We did not test $g_1$ in $y^8$ and $x\,y^7$. We skip testing $g_1$.
    \end{enumerate}
  \item[] For the monomial $x^2\,y^7$
    \begin{enumerate}
    \item[] We did not test $g_1$ in $x\,y^7$ and $x^2\,y^6$. We skip testing $g_1$.
    \end{enumerate}
  \item[] For the monomial $x^3\,y^6$
    \begin{enumerate}
    \item[] We did not test $g_1$ in $x^2\,y^6$ and $x^3\,y^5$. We skip testing $g_1$.
    \end{enumerate}
  \item[] For the monomial $x^4\,y^5$
    \begin{enumerate}
    \item[] We did not test $g_1$ in $x^3\,y^5$ and $x^4\,y^4$. We skip testing $g_1$.
    \end{enumerate}
  \item[] For the monomial $x^5\,y^4$
    \begin{enumerate}
    \item[] We did not test $g_1$ in $x^4\,y^4$. We skip testing $g_1$.
    \item[] The relation $g_2=x^5$ succeeds.
    \end{enumerate}
  \item[] For the monomial $x^6\,y^3$
    \begin{enumerate}
    \item[] Should the relation $g_1=y^2$ fail, we would have to add
      $y^2$ and $x^6\,y$ to the staircase, raising its size to $15$.
      We skip testing $g_1$.
    \item[] Should the relation $g_2=x^5$ fail, we would have to add
      $x^5$ and $x\,y^3$ to the staircase, raising its size to $15$.
      We skip testing $g_2$.
    \end{enumerate}
  \item[] For the monomial $x^7\,y^2$
    \begin{enumerate}
    \item[] The relation $g_1=y^2$ succeeds.
    \item[] The relation $g_2=x^5$ succeeds.
    \end{enumerate}
  \item[] For the monomial $x^8\,y$
    \begin{enumerate}
    \item[] The relation $g_2=x^5$ succeeds.
    \end{enumerate}
  \item[] For the monomial $x^9$
    \begin{enumerate}
    \item[] The relation $g_2=x^5$ succeeds.
    \end{enumerate}
  \item[] The algorithm returns relations
    $y^2,x^5$, the first one with a shift $x^7$
    and the other one with a shift $x^4$.
  \end{enumerate}
  The following figure shows the visited monomials where at least one
  relation was tested ($\cdot$) and those completely skipped
  ($\times$).
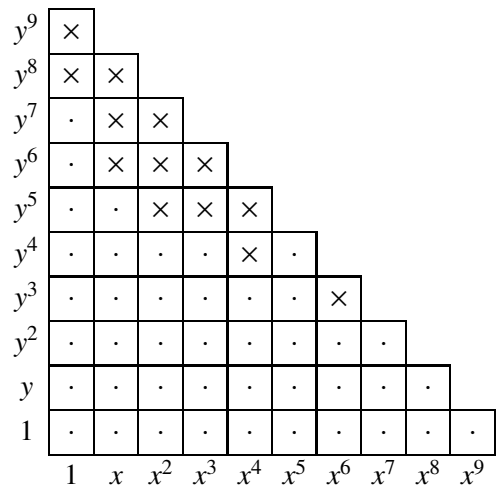
  \[
  \begin{ytableau}
    \none[y^9] &\times\\
    \none[y^8] &\times &\times\\
    \none[y^7] &\cdot &\times &\times\\
    \none[y^6] &\cdot &\times &\times &\times\\
    \none[y^5] &\cdot &\cdot  &\times &\times &\times\\
    \none[y^4] &\cdot &\cdot  &\cdot  &\cdot  &\times &\cdot\\
    \none[y^3] &\cdot &\cdot  &\cdot  &\cdot  &\cdot  &\cdot &\times\\
    \none[y^2] &\cdot &\cdot  &\cdot  &\cdot  &\cdot  &\cdot &\cdot &\cdot\\
    \none[y]   &\cdot &\cdot  &\cdot  &\cdot  &\cdot  &\cdot &\cdot &\cdot &\cdot\\
    \none[1]   &\cdot &\cdot  &\cdot  &\cdot  &\cdot  &\cdot &\cdot &\cdot &\cdot &\cdot\\
    \none      &\none[1]    &\none[x]  &\none[x^2] &\none[x^3] &\none[x^4]
               &\none[x^5] &\none[x^6] &\none[x^7] &\none[x^8] & \none[x^9]
  \end{ytableau}
  \]
  Although the monomial $x^4\,y^4$ was completely skipped, $x^5\,y^4$
  is not thanks to the relation $x^5$ that must be tested.
\end{example}


\section{The Adaptive version of the \sFGLM algorithm}
\label{s:asFGLM}
While the \BMS and \aBMS algorithms are iterative algorithms, the
\sFGLM algorithm is global,
see~\cite{issac2015,berthomieu:hal-01253934}
and~\cite[Section~4]{part1}. It finds the \gb of the ideal of
relations by computing the column rank profile of a big multi-Hankel
matrix indexed by a set of monomials $T$. In practice, this set $T$
must contain all the monomials less than the monomials in the \gb of
relations.

To circumvent the inherent complexity of computing the rank profile
of a big multi-Hankel matrix, the authors proposed an adaptive
algorithm behaving more closely to the \FGLM algorithm,
see~\cite{FGLM93}.

The goal is to iterate on a monomial $t$ and compute, for a set
$S$ such that $H_{S,S}$ is full rank, if $H_{S\cup\{t\},S\cup\{t\}}$ is
also full rank. If it is, then $t$ is added to $S$, otherwise a
relation with support in $S\cup\{t\}$ has been found. No further
relation with leading term a multiple of $t$ will be computed. When a
given lower-bound on the size of the staircase is reached, the
algorithm stops and computes the remaining relations from the leading
terms lying on the border of the staircase.

This yields the \asFGLM algorithm: Algorithm~\ref{algo:adapt_sfglm}.

\begin{algorithm2e}[htbp!]\label{algo:adapt_sfglm}
  \small
  \DontPrintSemicolon
  \TitleOfAlgo{\asFGLM (simple version).}
  \KwIn{A table $\bu=(u_{\bi})_{\bi\in\N^n}$ with coefficients in
    $\K$, $\prec$ a monomial ordering and $d$ a given bound.}
  \KwOut{A reduced truncated \gb of a zero-dimensional ideal of degree
    $\geq d$.}
  $L:=\{1\}$.\tcp*{set of next terms to study}
  $S:=\emptyset$.\tcp*{the useful staircase with respect to $\prec$}
  $G:=\emptyset,G':=\emptyset$.\;
  \While{$L\neq\emptyset$}{
    $t:=\min_{\prec}(L)$.\;
    \uIf{$H_{S\cup\{t\},S\cup\{t\}}$ is full rank}{
      $S:=S\cup\{t\}$ and $L:=L\cup\left\{x_i\,t,
        i=1,\ldots,n\right\}\setminus\{t\}$.\;
      Remove multiples of elements of $G'$ in $L$.\;
      \If(\tcp*[f]{early termination}){$\#\,S\geq d$}{
        \While{$L\neq\emptyset$}{
          $t':=\min_{\prec}(L)$.\;
          Find $\balpha$ such that $H_{S,S}\,\balpha + H_{S,\{t'\}}=0$.\;
          $G:=G\cup\left\{t'+\sum_{s\in S}\alpha_s\,s\right\}$.\;
          Remove multiples of elements of $t'$ in $L$.\;
        }
        \KwRet $G$.
      }
    }
    \Else{
      Find $\balpha$ such that $H_{S,S}\,\balpha + H_{S,\{t\}}=0$.\;
      $G':=G'\cup\{t\}$.\;
      $G:=G\cup\left\{t+\sum_{s\in S}\alpha_s\,s\right\}$.\;
      Remove multiples of $t$ in $L$ and sort $L$ by increasing order.
    }
  }
  \KwErr ``Run \sFGLM''.
\end{algorithm2e}

\begin{example}\label{ex:asfglm}
  We give the trace of the algorithm on the sequence
  $\bu=(2^i\,3^j\,(i+1))_{(i,j)\in\N^2}$ with the $\DRL(y\prec x)$
  ordering with a lower bound $2$ on the staircase size.
  \begin{enumerate}
  \item[] We set $L=\{1\}$, $S=\emptyset$, $G'=\emptyset$.
  \item[] We set $t=1$ and build the matrix $H_{S\cup\{1\},S\cup\{1\}}=\pare{
      \begin{smallmatrix}1
      \end{smallmatrix}}$ that is full rank. Hence $S=\{1\}$ and
    $L=\{y,x\}$.
  \item[] We set $t=y$ and build the matrix $H_{S\cup\{y\},S\cup\{y\}}=\pare{
      \begin{smallmatrix}1 &3\\3 &9        
      \end{smallmatrix}}$ that is not full rank. Solving
    $H_{S,S}\,\balpha+H_{S,\{y\}}=0$ yields relation $y-3$, so
    $G=\{y-3\},G'=\{y\}$ and $L$ is updated to $\{x\}$.
  \item[] We set $t=x$ and  build the matrix $H_{S\cup\{x\},S\cup\{x\}}=\pare{
      \begin{smallmatrix}1 &4\\4 &12
      \end{smallmatrix}}$ that is full rank. Hence $S=\{1,x\}$ and
    $L=\{x^2\}$.
  \item[] Now $\#\,S$ is greater or equal to the bound $2$.
    Solving $H_{S,S}\,\balpha+H_{S,\{x^2\}}=0$ yields relation $x^2-4\,x+4$, so
    $G=\{y-3,x^2-4\,x+4\}$ and $L$ is updated to $\emptyset$.

    Furthermore, the relation $y-3$ has been tested with a shift
    $\{1,y\}$ while
    the relation $x^2-4\,x+4$ has been tested with a shift $\{1,x\}$.
  \end{enumerate}
\end{example}

\begin{remark}\label{rk:asfglm_lowerbound}
  If no lower bound on the size of $S$ were given, then an infinite loop
  might occur on a non linear recurrent sequence. For instance, on the
  factorial  sequence $(i!)_{i\in\N}$, all the monomials $x^i$ would be found
  in the staircase.

  If we know the sequence is linear recurrent, then we can remove this
  bound. In that case, the last step of Example~\ref{ex:asfglm}
  becomes:
  \begin{enumerate}
  \item[] We set $t=x^2$ and build the matrix
    $H_{S\cup\{x^2\},S\cup\{x^2\}}=\pare{
      \begin{smallmatrix}1 &4 &12\\4 &12 &32\\12 &32 &80        
      \end{smallmatrix}}$ that is not full rank. Solving
    $H_{S,S}\,\balpha+H_{S,\{x^2\}}=0$ yields relation $x^2-4\,x+4$, so
    $G=\{y-3,x^2-4\,x+4\},G'=\{y,x^2\}$ and $L$ is updated to $\emptyset$.
  \end{enumerate}
  Furthermore, the relation $y-3$ has been tested with shift $\{1,y\}$ while
  the relation $x^2-4\,x+4$ has been tested with a shift $\{1,x,x^2\}$.
\end{remark}

For a generic sequence, the algorithm computes the ideal of
relations of the sequence. However, it is easy to make a
sequence such that the algorithm fails. It suffices to have a sequence
whose staircase $S$ has a subset $S'$ such that
the matrix $H_{S',S'}$ has a rank defect.

This motivated the authors to extend the algorithm to bypass this
issue in~\cite{berthomieu:hal-01253934}.

We give an example of what can happen when the wrong relations are
computed and describe their shifts.
\begin{example}\label{ex:fail_asfglm}
  We consider the ideal $I=\langle
  y^2-y,x^2\,y-x\,y,x^4-6\,x^3+11\,x^2-6\,x\rangle\subseteq\F_{11}[x,y]$
  and a sequence
  $\bu=(u_{i,j})_{(i,j)\in\N^2}$ over $\F_{11}$
  made from this ideal and some initial conditions. The first terms of
  the sequence are $\pare{
    \begin{smallmatrix}
      1 &2 &2 &2 &2 &\cdots\\
      3 &4 &4 &4 &4 &\cdots\\
      3 &4 &4 &4 &4 &\cdots\\
      -1 &4 &4 &4 &4&\cdots\\
      1 &4 &4 &4 &4&\cdots\\
      \vdots &\vdots &\vdots &\vdots &\vdots &\ddots
    \end{smallmatrix}}$.
  We also call the algorithm on the $\DRL(y\prec x)$ ordering.
  \begin{enumerate}
  \item[] We set $L=\{1\}$, $S=\emptyset$, $G'=\emptyset$.
  \item[] We set $t=1$ and build the matrix $H_{S\cup\{1\},S\cup\{1\}}=\pare{
      \begin{smallmatrix}1
      \end{smallmatrix}}$ that is full rank. Hence $S=\{1\}$ and
    $L=\{y,x\}$.
  \item[] We set $t=y$ and build the matrix $H_{S\cup\{y\},S\cup\{y\}}=\pare{
      \begin{smallmatrix}1 &2\\2 &2        
      \end{smallmatrix}}$ that is full rank. Hence $S=\{1,y\}$ and
    $L=\{x,y^2,x\,y\}$.
  \item[] We set $t=x$ and build the matrix $H_{S\cup\{x\},S\cup\{x\}}=\pare{
      \begin{smallmatrix}1 &2 &3\\2 &2 &4\\3 &4 &3        
      \end{smallmatrix}}$ that is full rank. Hence $S=\{1,y,x\}$ and
    $L=\{y^2,x\,y,x^2\}$.
  \item[] We set $t=y^2$ and build the matrix $H_{S\cup\{y^2\},S\cup\{y^2\}}=\pare{
      \begin{smallmatrix}1 &2 &3 &2\\2 &2 &4 &2\\3 &4 &3 &4\\2 &2 &4 &2        
      \end{smallmatrix}}$ that is not full rank. Solving
    $H_{S,S}\,\balpha+H_{S,\{y^2\}}=0$ yields relation $y^2-y$ with a
    shift $\{1,y,x,y^2\}$, so
    $G=\{y^2-y\},G'=\{y^2\}$ and $L$ is updated to $\{x\,y,x^2\}$.
  \item[] We set $t=x\,y$ and  build the matrix
    $H_{S\cup\{x\,y\},S\cup\{x\,y\}}=\pare{
      \begin{smallmatrix}1 &2 &3 &4\\2 &2 &4 &4\\3 &4 &3 &4\\4 &4 &4 &4        
      \end{smallmatrix}}$ that is not full rank. Solving
    $H_{S,S}\,\balpha+H_{S,\{x\,y\}}=0$ yields relation $x\,y-x-y+1$
    with a shift $\{1,y,x,x\,y\}$, so
    $G=\{y^2-y,x\,y-x-y+1\},G'=\{y^2,x\,y\}$ and $L$ is updated to $\{x^2\}$.
  \item[] We set $t=x^2$ and  build the matrix $H_{S\cup\{x^2\},S\cup\{x^2\}}=\pare{
      \begin{smallmatrix}1 &2 &3 &3\\2 &2 &4 &4\\
        3 &4 &3 &-1\\3 &4 &-1 &1        
      \end{smallmatrix}}$ that is full rank. Hence $S=\{1,y,x,x^2\}$ and
    $L=\{x^3\}$.
  \item[] We set $t=x^3$ and  build the matrix $H_{S\cup\{x^3\},S\cup\{x^3\}}=\pare{
      \begin{smallmatrix}1 &2 &3 &3 &-1\\
        2 &2 &4 &4 &4\\
        3 &4 &3 &-1 &1\\
        3 &4 &-1 &1 &2\\
        -1 &4 &1 &2 &6
      \end{smallmatrix}}$ that is not full rank. Solving
    $H_{S,S}\,\balpha+H_{S,\{x^3\}}=0$ yields relation
    $g_3=x^3+3\,x^2+10\,x+y+4$ with a shift
    $\{1,y,x,x^2\}$, so
    $G=\{y^2-y,x\,y-x-y+1,x^3+3\,x^2+10\,x+y+4\},G'=\{y^2,x\,y,x^3\}$
    and $L$ is updated to $\emptyset$.
  \end{enumerate}
  
  We can notice that
  \begin{itemize}
  \item the first relation, $y^2-y$ is really a relation of $\bu$
    but has only, a priori, a shift $\{1,y,x,y^2\}$, \ie its
    shift is $y^2$.
  \item the second relation, $x\,y-x-y+1$, is not a real
    relation of $\bu$ and is known to have a shift
    $\{1,y,x,x\,y\}$. Actually we can check that
    $[y^2\,(x\,y-x-y+1)]=0$ and $[x^2\,(x\,y-x-y+1)]=4$, so
    that the relation has a shift $\{1,y,x,y^2,x\,y\}$, \ie its
    shift is $x\,y$ and its fail is $x^3\,y$.
  \item the third relation, $x^3+3\,x^2+10\,x+y+4$,
    is not a true relation of $\bu$ and
    is known to have a shift
    $\{1,y,x,x^2,x^3\}$. Actually we can check that
    $[y^2\,(x\,y-x-y+1)]=0$ and
    $[x\,y\,(x^3+3\,x^2+10\,x+y+4)]=-1$, \ie its shift is
    $y^2$ and its fail is $x^4\,y$.
  \end{itemize}
  All in all, we computed the relation $x^3+3\,x^2+10\,x+y+4$
  assuming it should be
  valid when multiplied by $x^2$ or $x^3$, while it cannot be valid
  when multiplied by $x\,y\prec x^2\prec x^3$.
\end{example}


\section{Analogies and differences of the adaptive variants}
\label{s:comparison_adapt}
We now compare theoretically the \aBMS and the \asFGLM algorithms.
As the \aBMS algorithm
differs from the \BMS algorithm just in the execution: mainly some
testings are skipped, results from~\cite[Section~6]{part1} are still
valid for the \aBMS algorithm. On the other hand, the \asFGLM
algorithm does not
necessarily provide the same output as the \sFGLM algorithm.

\subsection{Closed staircase}
In~\cite[Section~5.1, Theorem~7]{part1}, we show that the \BMS
algorithm always returns a zero-dimensional ideal while the \sFGLM
algorithm can
return a zero-dimensional or a positive-dimensional ideal. This is in
fact one of the main differences between these two algorithms.

In the following theorem, we prove that the \aBMS algorithm and the
\asFGLM algorithm are closer on that matter assuming one knows the
size of the output staircase in advance.
\begin{theorem}\label{th:closed_staircase_adapt}
  Let $\bu$ be a sequence, $\prec$ be a monomial ordering and $d$ be
  the size of the staircase.
  
  Calling the \aBMS algorithm  on
  $\bu$, $\prec$, $d$ and a stopping monomial $M$ yields a
  truncated \gb of a zero-dimensional ideal.

  Calling the \asFGLM algorithms on $\bu$, $\prec$ and $d$
  yields a truncated \gb of a zero-dimensional ideal.
\end{theorem}
\begin{proof}
  The first part of the result comes directly from the line
  $G':=\Border(S')$ in the description of the \aBMS algorithm, Algorithm~\ref{algo:abms}.
  
  The second part of the result comes from the fact that the leading
  terms of the relations are lying in the border of the staircase and
  are minimal for both $\prec$ and $|$. Thus, for any variable $x_i$,
  there always exists a
  relation with leading term a pure power of $x_i$.
\end{proof}

It is possible to change this early termination procedure so that the
\asFGLM algorithm is closer to the \sFGLM algorithm, yielding a
potential positive-dimensional algorithm. If we still want to try to
close as much as possible the staircase with degenerate square
matrices,
it suffices to check that the
relation $t'+\sum_{s\in S}\alpha_s\,s$ is valid with a shift
$S\cup\{t'\}$. This yields Algorithm~\ref{algo:tweaked_adapt_sfglm}.
\begin{algorithm2e}[htbp!]\label{algo:tweaked_adapt_sfglm}
  \small
  \DontPrintSemicolon
  \TitleOfAlgo{Tweaked \asFGLM.}
  \KwIn{A table $\bu=(u_{\bi})_{\bi\in\N^n}$ with coefficients in
    $\K$, $\prec$ a monomial ordering and $d$ a given bound.}
  \KwOut{A reduced truncated \gb of a zero-dimensional ideal of degree
    $\geq d$.}
  $L:=\{1\}$.\tcp*{set of next terms to study}
  $S:=\emptyset$.\tcp*{the useful staircase with respect to $\prec$}
  $G:=\emptyset,G':=\emptyset$.\;
  \While{$L\neq\emptyset$}{
    $t:=\min_{\prec}(L)$.\;
    \uIf{$H_{S\cup\{t\},S\cup\{t\}}$ is full rank}{
      $S:=S\cup\{t\}$ and $L:=L\cup\left\{x_i\,t,
        i=1,\ldots,n\right\}\setminus\{t\}$.\;
      Remove multiples of elements of $G'$ in $L$.\;
      \If(\tcp*[f]{early termination}){$\#\,S\geq d$}{
        \While{$L\neq\emptyset$}{
          $t':=\min_{\prec}(L)$.\;
          Find $\balpha$ such that $H_{S,S}\,\balpha +
          H_{S,\{t'\}}=0$.\;
          \If{$H_{\{t'\},S}\,\balpha+H_{\{t'\},\{t'\}}=0$}{
            $G:=G\cup\left\{t'+\sum_{s\in S}\alpha_s\,s\right\}$.\;
          }
          Remove multiples of elements of $t'$ in $L$.\;
        }
        \KwRet $G$.
      }
    }
    \Else{
      Find $\balpha$ such that $H_{S,S}\,\balpha + H_{S,\{t\}}=0$.\;
      $G':=G'\cup\{t\}$.\;
      $G:=G\cup\left\{t+\sum_{s\in S}\alpha_s\,s\right\}$.\;
      Remove multiples of $t$ in $L$ and sort $L$ by increasing order.
    }
  }
  \KwErr ``Run \sFGLM''.
\end{algorithm2e}

\subsection{Reduction of relations}
The \asFGLM algorithm computes a staircase and then relations with support
in the staircase except their leading terms that lie on the border.
On the other hand, although the \aBMS algorithm may compute the same
ideal of relations as the \asFGLM algorithm, their \gb can be different.
\begin{theorem}\label{th:reduced_gb_adapt}
  Let $\bu$ be a sequence, $\prec$ be a monomial ordering and $d$ be
  the size of the staircase.
  
  Calling the \asFGLM algorithms on $\bu$, $\prec$, and $d$
  yields a truncated reduced \gb of an ideal.

  Calling the \aBMS algorithm  on
  $\bu$, $\prec$, $d$ and a stopping monomial $M$ yields a
  truncated minimal \gb of an ideal, which is not necessarily reduced.

  Furthermore, even if $\bu$ is linear recurrent and the
  \asFGLM algorithm computes the ideal of relations of $\bu$, then there
  is no reason for the output of the \aBMS algorithm to be reduced.
\end{theorem}
\begin{proof}
  For two distinct polynomials $g,g'$ in the \gb returned by
  \asFGLM algorithm, $\LT(g)$ does not divide any monomial in the
  support of $g'$. Hence the \gb is reduced.

  For two distinct polynomials $g,g'$ in the \gb returned by
  \aBMS algorithm, $\LT(g)$ does not divide $\LT(g')$. Hence the \gb
  is minimal. However, there is no reason for $\LT(g)$ not to divide
  any monomial in the support of $g'$.
\end{proof}
\begin{example}\label{ex:reduced}
  We let $\bu=\pare{i^2+j^2-1}_{(i,j)\in\N^2}$
  be a sequence and consider the $\DRL(y\prec x)$ ordering.
  The ideal of relations of $\bu$ is $I=\langle
  x\,y-x-y+1,x^2-y^2-2\,x+2\,y,y^3-3\,y^2+3\,y-1\rangle$.

  The \aBMS algorithm called on $\bu$ and the stopping monomial $y^5$ returns
  $g_1=x\,y-x-y+1$, with shift $x^2$,
  $g_2=x^2-\frac{1}{3}\,x\,y-y^2-\frac{5}{3}\,x+\frac{7}{3}\,y-\frac{1}{3}$,
  with shift $x^2$ and
  $g_3=y^3-\frac{1}{2}\,x\,y-3\,y^2+\frac{1}{2}\,x+\frac{7}{2}\,y-\frac{3}{2}$,
  with shift $y^2$. We can notice that $\{g_1,g_2,g_3\}$ is a \gb but not a reduced
  \gb of $I$.
  
  The \asFGLM algorithm called on $\bu$ and the set of all the
  monomials of degree at most $3$ yields relations
  $g_1'=x\,y-x-y+1,g_2'=x^2-y^2-2\,x+2\,y,g_3'=y^3-3\,y^2+3\,y-1$.
  We can notice that
  $\{g_1',g_2',g_3'\}=\{g_1,g_2+\frac{1}{3}\,g_1,g_3+\frac{1}{2}\,g_1\}$ 
  is a reduced \gb of $I$.


\end{example}

As for the \BMS algorithm, it is not hard to tweak the \aBMS algorithm
so that it returns a reduced \gb. It suffices to perform an
inter-reduction of the relations either at the end of each step of the
main \textbf{For} loop or just before returning the \gb,
see Algorithm~\ref{algo:tweaked_abms}.

\begin{algorithm2e}[htbp!]\label{algo:tweaked_abms}
  \small
  \DontPrintSemicolon
  \TitleOfAlgo{Tweaked \aBMS algorithm.}
  \KwIn{A table $\bu=(u_{\bi})_{\bi\in\N^n}$ with coefficients in
    $\K$, a monomial ordering $\prec$, a given bound $d$
    and a monomial $M$
    as the stopping condition.}
  \KwOut{A set $G$ of relations generating $I_M$.}
  $T := \{m\in\K[\x], m\preceq M\}$.\;
  $G := \{1\}$.\;
  $S := \emptyset$.\;
  \Forall{$m\in T$}{
    $S' := S$.\;
    \For{$g\in G$}{
      \If{$\LM(g)| m$}{
        \If{$\frac{m}{\LM(g)}\not\in\Stabilize(S)$ 
          \KwAnd $\#\,\Stabilize\pare{S\cup\acc{\LM(g),
              \frac{m}{\LM(g)}}}> d$}{
          \KwNext.\tcp*{skip this relation testing}
        }
        $e:=\cro{\frac{m}{\LM(g)}\,g}_{\bu}$\;
        \If{$e\neq 0$}{
          $S':=S'\cup
          \acc{\cro{\frac{g}{e},\frac{m}{\LM(g)}}}$.\;
        }
      }
    }
    $S':=
    \min_{\fail(h)\in S'}\acc{[h,\fail(h)/\LM(h)]}$.\;
    $G':= \Border (S')$.\;
    \For{$g'\in G'$}{
      Let $g\in G$ such that $\LM(g)|\LM(g')$.\; 
      \uIf{$\LM(g) \nmid m$}{
        $g':=\frac{\LM(g')}{\LM(g)}\,g$.\;
      }
      \uElseIf{$\exists\,h\in S,
        \frac{m}{\LM(g')} |\fail(h)$}{
        $g':= \frac{\LM(g')}{\LM(g)}\,g
        -\cro{\frac{m}{\LM(h)}\,h}_{\bu}\,
        \frac{\LM(g')\,\fail(h)}{m}\,h$.\;
      }
      \lElse{
        $g':=g$.
      }
    }
    $G := \InterReduce(G')$ \;
    $S := S'$.\;
  }  
  \KwRet $G$.
\end{algorithm2e}

\subsection{Validity of relations}\label{ss:validity}
One of the main differences between the \BMS and the \sFGLM algorithms
is the validity of the relations they return. Given a \gb returned by
both algorithms. Loosely speaking, the
\sFGLM algorithm will only ensure
that all the relations in the \gb have the same shifts while for the
\BMS algorithm, the
smaller the leading term of a relation is, the larger its shift is
computed. See~\cite[Theorem~19]{part1}.

Naturally, if the given upper bound on the size of the staircase to
the \aBMS algorithm is correct, then the shifts computed by the \aBMS
algorithm are the same as those computed by the \BMS algorithm.

In Examples~\ref{ex:asfglm} and~\ref{ex:fail_asfglm}, we can see that
the shifts computed by the \asFGLM algorithm are not all the
same. This is the main difference between the \sFGLM and the \asFGLM algorithms.

In fact, we prove in the following Theorem~\ref{th:valid_shift_adapt}
that the larger the leading
term of a computed relation, the larger its shift.
\begin{theorem}\label{th:valid_shift_adapt}
  Let $\bu$ be a sequence, $\prec$ be a monomial ordering and $d$ be
  the size of the output staircase $S$. Let $S_M=\{m\in S,\ m\prec M\}$.

  Calling the \aBMS algorithm on $\bu$, $\prec$, $d$ and a stopping monomial
  $M$ yields
  relations $g_1,\ldots,g_r$ and shifts $v_1,\ldots,v_r$ such that
  \[\forall\,i, 1\leq i\leq r,\quad v_i\,\LM(g_i)\preceq M\]
  and $g_i$ is valid with a shift $v_i$, potentially $0$.

  Calling the \asFGLM algorithm on $\bu$, $\prec$ and $d$
  yields relations $g_1',\ldots,g_{r'}'$ such that
  \[\forall\,i, 1\leq i\leq r',\quad \deg\LM(g_i')\leq d\]
  and $g_i'$ has a shift $S_{\LM(g_i')}\cup\{\LM(g_i')\}$ if
  $\LM(g_i')\succ\max_{\prec}(S)$ and $S$ otherwise.
\end{theorem}
\begin{proof}
  The first part is clear from the behavior of both the \BMS and the
  \aBMS algorithms.

  The second part comes from the fact that if $g_i'$, with $\LM(g_i')=t$
  is found before $S$ is
  completed, then it was because the matrix $H_{S^t\cup\{t\},S^t\cup\{t\}}$ had a
  rank default, where $S^t$ is the state for $S$ at loop
  $t$. Furthermore, $S^t=S_{\LM(g_i')}=S_t$.

  Otherwise, it is computed by solving
  $H_{S,S}\,\balpha+H_{S,\{t'\}}=0$ so that the relation has only been
  tested with a shift $S$.
\end{proof}
In a way, the behavior of the
\asFGLM algorithm is the opposite of the behaviors of
the \BMS and the \aBMS algorithms.

Furthermore, if one uses Algorithm~\ref{algo:tweaked_adapt_sfglm}
instead of the \asFGLM algorithm, then
each returned relation $g_i'$ has a shift $S_{\LM(g_i')}\cup\{\LM(g_i')\}$.
\begin{example}
  Let us consider the sequence $\bu=(F_{i+1})_{(i,j)\in\N^2}$, where
  $(F_i)_{i\in\N}$ is the Fibonacci sequence. Its ideal of relation is
  $\langle y-1,x^2-x-1\rangle$ so that its staircase has size $2$.

  Calling the \asFGLM algorithm on this sequence with this bound of
  the staircase makes us creating the matrices
  \begin{enumerate}
  \item[] $H_{\{1\},\{1\}}$, which is full rank, hence $1\in S$;
  \item[] $H_{\{1,y\},\{1,y\}}$, which is not full rank, hence the relation
    $y-1$ is found with a shift $\{1,y\}$;
  \item[] $H_{\{1,x\},\{1,x\}}$, which is full rank, hence $x\in S$.
  \end{enumerate}
  Now, the staircase is found so it remains to solve
  $H_{S,S}\,\balpha+H_{S,\{x^2\}}=0$ yielding the relation $x^2-x-1$
  with a shift $S$.
\end{example}

\subsection{Monomial ordering and Set of Terms}\label{ss:shape_position}
In this section, we study how both algorithms handle a monomial
ordering that is not a weighted degree ordering. The classical
specification of the \BMS algorithm are that the ordering must be a
weighted ordering. However, when running the \aBMS algorithm, the
upper bound on the staircase size makes us never visit monomials
of degree more than twice this size. Therefore, we can now use any
monomial ordering with the \aBMS algorithm by just enumerating, in
increasing order, all the monomials of degree less than
twice the upper bound.

This allows us to deal with ideal in shape position with both the
\aBMS and the \asFGLM algorithms.


\begin{theorem}\label{th:shape_position_adapt}
  Let $\bu$ be a linear recurrent sequence whose
  ideal of relation $I$ is in shape position for the
  $\LEX(x_n\prec\cdots\prec x_2\prec x_1)$ ordering, \ie there exist
  $g_n$ squarefree and $f_{n-1},\ldots,f_1\in\K[x_n]$ with $\deg
  g_n=d,\deg f_i<d$ such that 
  $I=\langle g_n(x_n),x_{n-1}-f_{n-1}(x_n),\ldots,x_1-f_1(x_n)\rangle$.
  
  Assuming no error
  is thrown in the execution of the \asFGLM algorithm called on $\bu$, $d$ and
  $\LEX(x_n\prec\cdots\prec x_2\prec x_1)$, then the ouput is $I$.

  Calling the \aBMS algorithm on $\bu$, $d$ and
  $\LEX(x_n\prec\cdots\prec x_2\prec x_1)$ yields $I$.
\end{theorem}
\begin{proof}
  Assuming no error is thrown during the execution of the \asFGLM
  algorithm, the staircase is incrementally updated from
  $\emptyset$ to $\acc{1,x_n,\ldots,x_n^{d-1}}$. Then,
  the staircase size is reached and the early termination procedure
  solves the system $H_{S,S}\,\balpha+H_{S,\{t\}}=0$ for
  $t\in\acc{x_n^d,x_{n-1},\ldots,x_1}$ yielding
  $g_n(x_n),x_{n-1}-f_{n-1}(x_n),\ldots, x_1-f_1(x_n)$.

  For the \aBMS algorithm, we visit every monomial of degree
  at most $2\,d-1$. The first relation, $g_n(x_n)$ is computed by the
  algorithm visiting monomials $1,x_n,\ldots,x_n^{2\,d-1}$ like the
  \BM algorithm. Then, each
  relation $x_i-f_i(x_n)$ is computed by visiting monomials
  $x_i,x_i\,x_n,\ldots,x_i\,x_n^{d-1}$, all of degree less than
  $2\,d-1$.
\end{proof}
\begin{example}
  We let $\bu=(F_{4\,i+k+1})_{(i,j,k)\in\N^3}$, where $(F_i)_{i\in\N}$
  is the Fibonacci sequence. The ideal of relations of $\bu$ is
  $I=\langle z^2-z-1,y-1,x-3\,z-2\rangle$ with a staircase of size $2$.

  For the \asFGLM called on $\bu$, $d=2$ and the
  $\LEX(z\prec y\prec x)$ ordering, the algorithm creates the matrices
  \begin{enumerate}
  \item[] $H_{\{1\},\{1\}}=\pare{
      \begin{smallmatrix}
        1
      \end{smallmatrix}}
    $, which is full rank, hence $1\in S$;
  \item[] $H_{\{1,z\},\{1,z\}}=\pare{
      \begin{smallmatrix}
        1 &1\\1 &2
      \end{smallmatrix}}$, which is full rank, hence $z\in S$.
  \end{enumerate}
  Now, the staircase is found so it remains to solve
  \begin{enumerate}
  \item[] $H_{S,S}\,\balpha+H_{S,\{z^2\}}=0$ yielding the relation
    $g_1=z^2-z-1$;
  \item[] $H_{S,S}\,\balpha+H_{S,\{y\}}=0$ yielding the relation
    $g_2=y-1$;
  \item[] $H_{S,S}\,\balpha+H_{S,\{x\}}=0$ yielding the relation
    $g_3=x-3\,z-2$.
  \item[] The algorithm returns $\langle g_1,g_2,g_3\rangle=I$.
  \end{enumerate}

  Calling the \aBMS algorithm on $\bu$, $d=2$, the stopping monomial
  $x\,z$ and
  $\LEX(z\prec y\prec x)$ ordering makes us visit the set of all monomials
  of degree at most $2\,d-1=3$ less than $x\,z$, \ie
  $\{1,z,z^2,z^3,y,y\,z,y\,z^2,y^2,y^2\,z,y^3,x,x\,z\}$.
  \begin{enumerate}
  \item[] The algorithms tests the relation $g=1$ in
    $u_{0,0,0}=F_1=1$ where it fails. It has now relations
    $g_1=x,g_2=y$ and $g_3=z$.
  \item[] Testing $g_3=z$ in $u_{0,0,2}=F_2=1$, it updates now the
    relation to $g_3=z-1$. Going on testing $g_3=z-1$ in
    $u_{0,0,2}=F_3=2$ and $u_{0,0,3}=F_4=3$, it is able to guess that
    $g_3=z^2-z-1$. The staircase is now $\{1,z\}$ of size $2$ so it
    has been found. As anticipated, there is no need to go further in
    that
    direction.
  \item[] Testing $g_2=y$ in $u_{0,1,0}=F_1=1$, the relation is updated
    to $g_2=y-1$.
  \item[] Then, it checks that this relation is valid in
    $u_{0,1,1}$ but skips $u_{0,1,2},u_{0,2,0},u_{0,2,1},u_{0,3,0}$
    thanks to its criterion.
  \item[] It remains to test $g_3=x$ in $u_{1,0,0}=F_5=5$. It fails and
    the algorithm updates the relation to $g_3=x-5$.
  \item[] Finally, $g_3=x-5$ is tested in $u_{1,0,1}=F_6=8$ and the
    relation is updated to $g_3=x-3\,z-2$.
  \item[] The algorithm returns $\langle g_1,g_2,g_3\rangle=I$.
\end{enumerate}
\end{example}


\section{Complexity and Benchmarks of the adaptive variants}
\label{s:implem_adapt}
In this section, we present some benchmarks to compare how
the \aBMS and the \asFGLM algorithms behave.

Four families of ideals of relations are used to make the
sequences.

\begin{itemize}
\item In the first family, 
  the leading monomials of the ideal of
  relations are $\langle y^{\floor{d/2}},x^d\rangle$. Thus, its staircase is a
  rectangle of size around $d^2/2$. In three variables,
  the leading monomials are
  $\langle z^{\ceil{d/3}},y^{\floor{d/2}},x^d\rangle$,
  so that the staircase is a rectangular cuboid of size
  around $d^3/6$. This family will be called \emph{Rectangle}.

\item In the second family,
  the leading monomials of the ideal of relations
  are $\langle x\,y,y^d,x^d\rangle$.  Thus, its staircase looks like a
  \textsc{L} and has size $2\,d-1$.
  In three variables, the leading monomials are
  $\langle y\,z,x\,z,x\,y,z^d,y^d,x^d\rangle$, so that the staircase has size
  $3\,d-2$. This family will be called \emph{\textsc{L}
    shape}. It was considered as the worst case
  in~\cite{issac2015,berthomieu:hal-01253934} for the \asFGLM
  algorithm
  for the number of queries.
  It should also be a worst case for
  the \aBMS algorithm.

\item In the third family, 
  the leading monomials of the ideal of relations
  are all the monomials of degree $d$. Thus, its staircase
  is a simplex and has size $\binom{d+1}{2}=\frac{d\,(d+1)}{2}$ in
  two variables. In three variables, the staircase has
  size $\binom{d+2}{3}=\frac{d\,(d+1)\,(d+2)}{6}$. This family
  will be called \emph{Simplex}. It should be the best case for both
  the \sFGLM and the \BMS algorithms.

\item In the last family, 
  the leading monomials of the ideal of relations
  are $\langle y^d,x\rangle$.  Thus, its staircase looks like a
  line and has size $d$.
  In three variables, the leading monomials are
  $\langle z^d,y,x\rangle$, so that the staircase has also size
  $d$. This is the generic family for a $\LEX(z\prec y\prec x)$ basis
  and this example corresponds to the change of ordering application,
  see Section~\ref{ss:shape_position}.
  This family will be called \emph{Shape position}.
\end{itemize}
For the first three families,
we called the algorithms with the $\DRL(z\prec y\prec x)$ ordering, for the
last one, we called them with the $\LEX (z\prec y\prec x)$ ordering.

For the \aBMS algorithm, we used
Proposition~\ref{prop:upperbound_repeat} to
estimate sharply the stopping monomial.

\subsection{Counting the number of table queries}
The \asFGLM algorithm computes all the multi-Hankel matrices whose rows and columns
are all the terms that are in the staircase or are a leading monomial in the
\gb.

Likewise, the \aBMS algorithm needs to test each relation, with support
in $S\cup\LM(\cG)$, shifted by as many monomial as in $S$.

Therefore, we have the following proposition.

\begin{proposition}\label{prop:asfglm_queries}
  Let $\bu=(u_{\bi})_{\bi\in\N^n}$
  be a sequence and $\cG$ be a reduced \gb of its ideal of
  relations for a total degree ordering.
  
  Let $S$ be the staircase of $\cG$, $S^+=S\cup\LM(\cG)$. Let
  $S+T=\{s\,t,\ s\in S,t\in T\}$ and
  $2\,S=S+S=\{s\,s',\ s,s'\in S\}$.
  
  Let $d_S$ be the greatest
  degree of the elements in $S$, $d_{\cG}$ be the
  greatest degree of the elements in $\cG$ and
  $d_{\max}=\max(d_S,d_{\cG})$.

  Let $\cS(d)$ be the
  simplex of all monomials of degree $d$.

  Then, the \aBMS algorithm needs to perform at least $\#\,(S+S^+)$
  and at most $\#\,\cS(d_S+d_{\max})=\binom{n+d_S+d_{\max}}{n}$ queries to
  the sequence.

  The \asFGLM algorithm needs to perform at least $\#\,(2\,\,S)$ and
  fewer than $\#\,(2\,S^+)$ queries to $\bu$.
  In the worst case, this number grows as $(\#\,S^+)^2$.
\end{proposition}
\begin{figure}[htbp!]
  \pgfplotsset{
    small,
    width=12cm,
    height=7cm,
    legend cell align=left,
    legend columns=5,
    legend style={at={(-0.05,0.98)},anchor=south
      west,font=\scriptsize,
    }
  }
  \centering
  \begin{tikzpicture}[baseline]
    \begin{axis}[
      ymode=log, xlabel={$d$}, xlabel
      style={at={(0.95,0.1)}},
      xmin=3.8,xmax=25.2,ymin=1.8,ymax=18,
      xtick={2,...,25},
      ytick={1,2,3,4,5,6,7,8,9,10,20,30,40,50,60,70,80,90,
        100,200,300,400,500,600,700,800,900,1000},
      yticklabels={},
      extra y ticks={1,5,10,50,100,500},
      extra y tick labels={1,5,10,50,100,500},
      ylabel={\#\,Queries/\#\,S},
      ylabel style={at={(0.08,0.750)}}
      ]
      \addlegendimage{empty legend}
      \addlegendentry{}
      \addlegendimage{legend image with text=Rectangle}
      \addlegendentry{}
      \addlegendimage{legend image with text=\textsc{L} shape}
      \addlegendentry{}
      \addlegendimage{legend image with text=Simplex}
      \addlegendentry{}
      \addlegendimage{legend image with text=Shape position}
      \addlegendentry{}
      
      \addlegendimage{empty legend}
      \addlegendentry{\asFGLM}
      \addplot[thick,every mark/.append style={solid},
      mark=triangle*,dashed,red,mark phase=1,mark repeat=4]
      plot coordinates {
        (4,25/8) (5,31/10) 
        (6,60/18) (7,70/21) (8,110/32) (9,124/36) 
        (10,176/50) (11,194/55) (12,258/72) (13,280/78) 
        (14,356/98) (15,382/105) (16,470/128) (17,500/136) 
        (18,600/162) (19,634/171) (20,746/200) 
        (21,784/210) (22,908/242) (23,950/253) 
        (24,1086/288) (25,1132/300) 
      };
      \addlegendentry{}
      \addplot[thick,every mark/.append style={solid},
      mark=triangle*,dashed,blue,mark phase=2,mark repeat=4]
      plot coordinates {
        (2,10/3) (3,19/5) (4,30/7) (5,43/9) 
        (6,58/11) (7,75/13) (8,94/15) (9,115/17) 
        (10,138/19) (11,163/21) (12,190/23) (13,219/25) 
        (14,250/27) (15,283/29) (16,318/31) (17,355/33) 
        (18,394/35) (19,435/37) (20,478/39)  (21,523/41) 
        (22,570/43) (23,619/45) (24,670/47) (25,723/49)
      };
      \addlegendentry{} 
      \addlegendimage{empty legend}
      \addlegendentry{}
      \addplot[thick,every mark/.append style={solid},
      mark=triangle*,dashed,orange,mark phase=4,mark repeat=4]
      plot coordinates {
        (2,6/2) (3,9/3) (4,11/4) (5,13/5) 
        (6,15/6) (7,17/7) (8,19/8) (9,21/9) 
        (10,23/10) (11,25/11) (12,27/12) (13,29/13) 
        (14,31/14) (15,33/15) (16,35/16) (17,37/17) 
        (18,39/18) (19,41/19) (20,43/20) (21,45/21) 
        (22,47/22) (23,49/23) (24,51/24) (25,53/25)
      };
      \addlegendentry{}

      \addlegendimage{empty legend}
      \addlegendentry{\aBMS}
      \addplot[thick,every mark/.append style={solid,rotate=180},
      mark=triangle*,dotted,red,mark phase=1,mark repeat=4]
      plot coordinates {
        (4,41/8) (5,58/10) 
        (6,103/18) (7,132/21) (8,198/32) (9,236/36) 
        (10,320/50) (11,371/55) (12,478/72) (13,541/78) 
        (14,663/98) (15,731/105) (16,882/128) (17,967/136) 
        (18,1141/162) (19,1238/171) (20,1418/200) 
      };
      \addlegendentry{}
      \addplot[thick,every mark/.append style={solid,rotate=180},
      mark=triangle*,dotted,blue,mark phase=2,mark repeat=4]
      plot coordinates {
        (2,10/3) (3,21/5) (4,36/7) (5,55/9) 
        (6,78/11) (7,105/13) (8,136/15) (9,171/17) 
        (10,210/19) (11,253/21) (12,300/23) (13,351/25) 
        (14,406/27) (15,465/29) (16,528/31) (17,595/33) 
        (18,666/35) (19,741/37) (20,820/39) (21,903/41) 
        (22,990/43) (23,1081/45) (24,1176/47) (25,1275/49)
      };
      \addlegendentry{}
      \addlegendimage{empty legend}
      \addlegendentry{}
      \addplot[thick,every mark/.append style={solid,rotate=180},
      mark=triangle*,dotted,orange,mark phase=4,mark repeat=4]
      plot coordinates {
        (2,6/2) (3,10/3) (4,14/4) (5,18/5) 
        (6,22/6) (7,26/7) (8,30/8) (9,34/9) 
        (10,38/10) (11,42/11) (12,46/12) (13,50/13) 
        (14,54/14) (15,58/15) (16,62/16) (17,66/17) 
        (18,70/18) (19,74/19) (20,78/20) (21,82/21) 
        (22,86/22) (23,90/23) (24,94/24) (25,98/25)
      };
      \addlegendentry{}
      
      \addlegendimage{empty legend}
      \addlegendentry{Both algorithms}
      \addlegendimage{empty legend}
      \addlegendentry{}
      \addlegendimage{empty legend}
      \addlegendentry{}
      \addplot[thick,every mark/.append style={solid},
      mark=square*,green,mark phase=3,mark repeat=4]
      plot coordinates {
        (2,10/3) (3,21/6) (4,36/10) (5,55/15) 
        (6,78/21) (7,105/28) (8,136/36) (9,171/45) 
        (10,210/55) (11,253/66) (12,300/78) (13,351/91) 
        (14,406/105) (15,465/120) (16,528/136) (17,595/153) 
        (18,666/171) (19,741/190) (20,820/210) (21,903/231) 
        (22,990/253) (23,1081/276) (24,1176/300) (25,1275/325) 
      };
      \addlegendentry{}
    \end{axis}
  \end{tikzpicture}
  \caption{Number of table queries (\textsc{2D}): \asFGLM \& \aBMS}
  \label{fig:queries2Dadapt}
\end{figure}

In the experiments of Figures~\ref{fig:queries2Dadapt}
and~\ref{fig:queries3Dadapt},
we can see that for the Rectangle family, the \asFGLM algorithm perform much
fewer
queries than the \aBMS.

For the \textsc{L} shape family, the size of the staircase only grows as
$O(d)$. Our experiments suggest that the number of queries grows as
$O(d^n)$ for the \aBMS algorithm, while it only grows as $O(d^2)$ for the
\asFGLM algorithm. This can be a huge advantage in dimension at least $3$.

We can see that the \aBMS algorithm
cannot take profit from the size of the staircase in the
\textsc{L} shape family as it needs as many queries as in the
Simplex family. Yet, although the \textsc{L} shape family is a worst
case for the
\asFGLM algorithm, it is still able to query fewer sequence terms for
the \textsc{L} shape family than for the Simplex family.

\begin{figure}[htbp!]
  \pgfplotsset{
    small,
    width=12cm,
    height=7cm,
    legend cell align=left,
    legend columns=5,
    legend style={at={(-0.05,0.98)},anchor=south
      west,font=\scriptsize,
    }
  }
  \centering
  \begin{tikzpicture}[baseline]
    \begin{axis}[
      ymode=log, xlabel={$d$}, xlabel
      style={at={(0.95,0.1)}},
      xmin=3.7,xmax=15.2,ymin=3.8,ymax=65,
      xtick={2,...,15},
      ytick={1,2,3,4,5,6,7,8,9,10,20,30,40,50,60,70,80,90,
        100,200,300,400,500,600,700,800,900,
        1000,2000,3000,4000,5000,6000,7000,8000,9000,
        10000,20000},
      yticklabels={},
      extra y ticks={1,5,10,50,100,500,1000,5000,10000},
      extra y tick labels={1,5,10,50,100,500,1000,5000,10000},
      ylabel={\#\,Queries/\#\,S},
      ylabel style={at={(0.08,0.75)}}
      ]
      \addlegendimage{empty legend}
      \addlegendentry{}
      \addlegendimage{legend image with text=Rectangle}
      \addlegendentry{}
      \addlegendimage{legend image with text=\textsc{L} shape}
      \addlegendentry{}
      \addlegendimage{legend image with text=Simplex}
      \addlegendentry{}
      \addlegendimage{legend image with text=Shape position}
      \addlegendentry{}
      
      \addlegendimage{empty legend}
      \addlegendentry{\asFGLM}
      \addplot[thick,every mark/.append style={solid},
      mark=triangle*,dashed,red,mark phase=1,mark repeat=4]
      plot coordinates {
        (4,72/16) (5,90/20) 
        (6,174/36) (7,334/63) (8,534/96) (9,604/108) 
        (10,1206/200) (11,1332/220) 
        (12,1780/288) (13,2484/390) (14,3168/490) (15,3402/525)
      };
      \addlegendentry{}
      \addplot[thick,every mark/.append style={solid},
      mark=triangle*,dashed,blue,mark phase=2,mark repeat=4]
      plot coordinates {
        (2,20/4) (3,42/7) (4,69/10) (5,102/13) 
        (6,141/16) (7,186/19) (8,237/22) (9,294/25) 
        (10,357/28) (11,426/31) (12,501/34) (13,582/37) 
        (14,669/40) (15,762/43) 
      };
      \addlegendentry{}
      \addlegendimage{empty legend}
      \addlegendentry{}
      \addplot[thick,every mark/.append style={solid},
      mark=triangle*,dashed,orange,mark phase=4,mark repeat=4]
      plot coordinates {
        (2,8/2) (3,13/3) (4,18/4) (5,21/5) 
        (6,25/6) (7,30/7) (8,33/8) (9,36/9) 
        (10,40/10) (11,46/11) (12,49/12) (13,52/13) 
        (14,55/14) (15,59/15)
      };
      \addlegendentry{}
     
      \addlegendimage{empty legend}
      \addlegendentry{\aBMS}
      \addplot[thick,every mark/.append style={solid,rotate=180},
      mark=triangle*,dotted,red,mark phase=1,mark repeat=4]
      plot coordinates {
        (4,207/16) (5,330/20) 
        (6,709/36) (7,1265/63) (8,2159/96) (9,2739/108) 
        (10,4819/200)
      };
      \addlegendentry{}
      \addplot[thick,every mark/.append style={solid,rotate=180},
      mark=triangle*,dotted,blue,mark phase=2,mark repeat=4]
      plot coordinates {
        (2,20/4) (3,56/7) (4,120/10) (5,218/13) 
        (6,357/16) (7,545/19) (8,784/22) (9,1090/25) 
        (10,1457/28) (11,1907/31) (12,2424/34) (13,3043/37) 
        (14,3741/40) (15,4557/43)
      };
      \addlegendentry{}
      \addlegendimage{empty legend}
      \addlegendentry{}
      \addplot[thick,every mark/.append style={solid,rotate=180},
      mark=triangle*,dotted,orange,mark phase=4,mark repeat=4]
      plot coordinates {
        (2,8/2) (3,13/3) (4,18/4) (5,23/5) 
        (6,28/6) (7,33/7) (8,38/8) (9,43/9) 
        (10,48/10) (11,53/11) (12,58/12) (13,63/13) 
        (14,68/14) (15,73/15) 
      };
      \addlegendentry{}
      
      \addlegendimage{empty legend}
      \addlegendentry{Both algorithms}
      \addlegendimage{empty legend}
      \addlegendentry{}
      \addlegendimage{empty legend}
      \addlegendentry{}
      \addplot[thick,every mark/.append style={solid},
      mark=square*,green,mark phase=3,mark repeat=4]
      plot coordinates {
        (2,20/4) (3,56/10) (4,120/20) (5,220/35) 
        (6,364/56) (7,560/84) (8,816/120) (9,1140/165) 
        (10,1540/220) (11,2024/286) (12,2600/364) (13,3276/455) 
        (14,4060/560) (15,4960/680)
      };
      \addlegendentry{}
    \end{axis}
  \end{tikzpicture}
  \caption{Number of table queries (\textsc{3D}): \asFGLM \& \aBMS}
  \label{fig:queries3Dadapt}
\end{figure}

\subsection{Counting the number of basic operations}
The complexity of the \BMS algorithm has been studied
in~\cite{Sakata09} yielding the following proposition.
\begin{proposition}\label{prop:bms_basicop}
  Let $\bu=(u_{\bi})_{\bi\in\N^n}$
  be a sequence, $\cG$ be a minimal \gb of its ideal of
  relations for a total degree ordering and $S$ be the staircase of $\cG$.

  Then, the \BMS algorithm performs at most
  $O\pare{(\#\,S)^2\,\LM(\cG)}$ operations to recover the
  ideal of relations of $\bu$.  
\end{proposition}

Obviously, the bound of Proposition~\ref{prop:bms_basicop} on the
number of basic operations applies to the \aBMS algorithm. Yet, since
the number of skipped relation testings is hard to predict, it is not
clear how to make it sharper for the \aBMS algorithm.

The \asFGLM computes the rank of a matrix of size at most
$\#\,S$. Furthermore, it solves as many linear systems with this matrix
as there are polynomials in the \gb. All in all, we have the following
result.
\begin{proposition}\label{prop:asfglm_basicop}
  Let $\bu=(u_{\bi})_{\bi\in\N^n}$
  be a sequence, $\cG$ be a reduced \gb of its ideal of
  relations for a total degree ordering and $S$ be the staircase of $\cG$.

  Then, the number of operations performed by the \asFGLM algorithm to
  recover the ideal of relations of $\bu$ is at most
  $O\pare{(\#\,S)^2\,(\#\,S+\#\,\LM(\cG))}$. 
\end{proposition}

In the following Figures~\ref{fig:basicop2Dadapt} and~\ref{fig:basicop3Dadapt},
we report on the ratio between the number of basic operations and the
cube of the size of the staircase.

\begin{figure}[htbp!]
  \pgfplotsset{
    small,
    width=12cm,
    height=7cm,
    legend cell align=left,
    legend columns=5,
    legend style={at={(-0.05,0.98)},anchor=south
      west,font=\scriptsize,
    }
  }
  \centering
  \begin{tikzpicture}[baseline]
    \begin{axis}[
      ymode=log,
      xlabel={$d$},
      xlabel style={at={(0.95,0.1)}},
      xmin=3.8,xmax=25.2,
      ymin=0.08,ymax=22,
      xtick={4,...,25},
      ytick={0.09,0.1,0.2,0.3,0.4,0.5,0.6,0.7,0.8,0.9,1,2,3,4,5,6,7,8,9,
        10,20,30,40,50,60,70,80,90,100,200,300,400,500,600,700,800,900,
      1000,2000,3000,4000,5000,6000,7000,8000,9000,10000,20000},
      yticklabels={},
      extra y ticks={0.1,0.5,1,5,10,50,100,500,1000,5000,10000},
      extra y tick labels={0.1,0.5,1,5,10,50,100,500,1000,5000,10000},
      ylabel={\#\,Basic Op/\#\,S$^3$},
      ylabel style={at={(0.08,0.75)}},
      ]
      \addlegendimage{empty legend}
      \addlegendentry{}
      \addlegendimage{legend image with text=Rectangle}
      \addlegendentry{}
      \addlegendimage{legend image with text=\textsc{L} shape}
      \addlegendentry{}
      \addlegendimage{legend image with text=Simplex}
      \addlegendentry{}
      \addlegendimage{legend image with text=Shape position}
      \addlegendentry{}

      \addlegendimage{empty legend}
      \addlegendentry{\asFGLM}
      \addplot[thick,every mark/.append style={solid},
      mark=triangle*,dashed,red,mark phase=1,mark repeat=4]
      plot coordinates {
        (4,210/8^3) (5,337/10^3) (6,1595/18^3) (7,2355/21^3)
        (8,7147/32^3) (9,9871/36^3) (10,24599/50^3) (11,32254/55^3)
        (12,69587/72^3) (13,87739/78^3) (14,170015/98^3)
        (15,208053/105^3)
        (16,371335/128^3) (17,443935/136^3) (18,742787/162^3)
        (19, 871616/171^3) (20,1384599/200^3) (21,1600259/210^3)
        (22,2436147/242^3) (23,2780359/253^3) 
        (24,4085075/288^3) (25,4613103/300^3)
      };
      \addlegendentry{}
      \addplot[thick,every mark/.append style={solid},
      mark=triangle*,dashed,blue,mark phase=2,mark repeat=4]
      plot coordinates {
        (2,52/3^3) (3,128/5^3) (4,215/7^3) (5,342/9^3) 
        (6,517/11^3) (7,748/13^3) (8,1043/15^3) (9,1410/17^3) 
        (10,1857/19^3) (11,2392/21^3) (12,3023/23^3) (13,3758/25^3) 
        (14,4605/27^3) (15,5572/29^3) (16,6667/31^3) (17,7898/33^3) 
        (18,9273/35^3) (19,10800/37^3) (20,12487/39^3) (21,14342/41^3) 
        (22,16373/43^3) (23,18588/45^3) (24,20995/47^3) (25,23602/49^3)
      };
      \addlegendentry{}
      \addplot[thick,every mark/.append style={solid},
      mark=triangle*,dashed,green,mark phase=3,mark repeat=4]
      plot coordinates {
        (2,52/3^3) (3,200/6^3) (4,615/10^3) (5,1610/15^3) 
        (6,3724/21^3) (7,7812/28^3) (8,15150/36^3) (9,27555/45^3) 
        (10,47520/55^3) (11,78364/66^3) (12,124397/78^3) (13,191100/91^3) 
        (14,285320/105^3) (15,415480/120^3) (16,591804/136^3) (17,826557/153^3) 
        (18,1134300/171^3) (19,1532160/190^3) (20,2040115/210^3) (21,2681294/231^3) 
        (22,3482292/253^3) (23,4473500/276^3) (24,5689450/300^3) (25,7169175/325^3)
      };
      \addlegendentry{}
      \addplot[thick,every mark/.append style={solid},
      mark=triangle*,dashed,orange,mark phase=4,mark repeat=4]
      plot coordinates {
        (2,22/2^3) (3,41/3^3) (4,59/4^3) (5,83/5^3) 
        (6,114/6^3) (7,153/7^3) (8,201/8^3) (9,259/9^3) 
        (10,328/10^3) (11,409/11^3) (12,503/12^3) (13,611/13^3) 
        (14,734/14^3) (15,873/15^3) (16,1029/16^3) (17,1203/17^3) 
        (18,1396/18^3) (19,1609/19^3) (20,1843/20^3) (21,2099/21^3) 
        (22,2378/22^3) (23,2681/23^3) (24,3009/24^3) (25,3363/25^3)
      };
      \addlegendentry{}

      \addlegendimage{empty legend}
      \addlegendentry{\aBMS}
      \addplot[thick,every mark/.append style={solid,rotate=180},
      mark=triangle*,dotted,red,mark phase=1,mark repeat=4]
      plot coordinates {
        (4,2115/8^3) (5,3108/10^3) (6,12589/18^3) (7,18206/21^3)
        (8,50410/32^3) (9,65117/36^3) (10,142499/50^3)
        (11,180157/55^3)
        (12,352288/72^3) (13,426269/78^3) (14,738578/98^3)
        (15,865913/105^3)
        (16,1432507/128^3) (17,1664165/136^3) (18,2567128/162^3)
        (19, 2926327/171^3) (20,4307996/200^3)
      };
      \addlegendentry{}
      \addplot[thick,every mark/.append style={solid,rotate=180},
      mark=triangle*,dotted,blue,mark phase=2,mark repeat=4]
      plot coordinates {
        (2,438/3^3) (3,1111/5^3) (4,2259/7^3) (5,3721/9^3) 
        (6,5817/11^3) (7,8417/13^3) (8,11375/15^3) (9,14881/17^3) 
        (10,19457/19^3) (11,23359/21^3) (12,29387/23^3) (13,34697/25^3) 
        (14,40977/27^3) (15,48195/29^3) (16,55495/31^3) (17,63243/33^3) 
        (18,72697/35^3) (19,80719/37^3) (20,91271/39^3) (21,101213/41^3) 
        (22,113021/43^3) (23,122455/45^3) (24,136731/47^3) (25,148925/49^3)
      };
      \addlegendentry{}
      \addplot[thick,every mark/.append style={solid,rotate=180},
      mark=triangle*,dotted,green,mark phase=3,mark repeat=4]
      plot coordinates {
        (2,427/3^3) (3,1759/6^3) (4,5241/10^3) (5,12860/15^3) 
        (6,27552/21^3) (7,53414/28^3) (8,95823/36^3) (9,161690/45^3) 
        (10,259672/55^3) (11,400330/66^3) (12,596325/78^3) (13,862630/91^3) 
        (14,1216684/105^3) (15,1678622/120^3) (16,2271453/136^3)
        (17,3021248/153^3) (18,3957182/171^3) (19,5112366/190^3)
        (20,6522843/210^3) (21,8228829/231^3) (22,10275367/253^3)
        (23,12709953/276^3) (24,15585915/300^3) (25,18960680/325^3)
      };
      \addlegendentry{}
      \addplot[thick,every mark/.append style={solid,rotate=180},
      mark=triangle*,dotted,orange,mark phase=4,mark repeat=4]
      plot coordinates {
        (2,136/2^3) (3,337/3^3) (4,686/4^3) (5,1108/5^3) 
        (6,1966/6^3) (7,2794/7^3) (8,3205/8^3) (9,4506/9^3) 
        (10,5598/10^3) (11,6283/11^3) (12,8114/12^3) (13,9678/13^3) 
        (14,10331/14^3) (15,12964/15^3) (16,15034/16^3) (17,16314/17^3) 
        (18,18661/18^3) (19,21618/19^3) (20,23282/20^3) (21,26869/21^3) 
        (22,29037/22^3) (23,31586/23^3) (24,35273/24^3) (25,39316/25^3)
      };
      \addlegendentry{}
    \end{axis}
  \end{tikzpicture}
  \caption{Number of basic operations (\textsc{2D}): \asFGLM \& \aBMS}
  \label{fig:basicop2Dadapt}
\end{figure}

\begin{figure}[htbp!]
  \pgfplotsset{
    small,
    width=12cm,
    height=7cm,
    legend cell align=left,
    legend columns=5,
    legend style={at={(-0.05,0.98)},anchor=south
      west,font=\scriptsize,
    }
  }
  \centering
  \begin{tikzpicture}[baseline]
    \begin{axis}[
      ymode=log,
      xlabel={$d$},
      xlabel style={at={(0.95,0.1)}},
      xmin=3.7,xmax=15.2,
      ymin=0.08,ymax=250,
      xtick={2,...,15},
      ytick={0.1,0.2,0.3,0.4,0.5,0.6,0.7,0.8,0.9,1,2,3,4,5,6,7,8,9,10,
        20,30,40,50,60,70,80,90,100,200,300,400,500,600,700,800,900,
        1000,2000,3000},
      yticklabels={},
      extra y ticks={0.1,0.5,1,5,10,50,100,500,1000},
      extra y tick labels={0.1,0.5,1,5,10,50,100,500,1000},
      ylabel={\#\,Basic Op/\#\,S$^3$},
      ylabel style={at={(0.08,0.75)}},
      ]
      \addlegendimage{empty legend}
      \addlegendentry{}
      \addlegendimage{legend image with text=Rectangle}
      \addlegendentry{}
      \addlegendimage{legend image with text=\textsc{L} shape}
      \addlegendentry{}
      \addlegendimage{legend image with text=Simplex}
      \addlegendentry{}
      \addlegendimage{legend image with text=ShapePosition}
      \addlegendentry{}

      \addlegendimage{empty legend}
      \addlegendentry{\asFGLM}
      \addplot[thick,every mark/.append style={solid},
      mark=triangle*,dashed,red,mark phase=1,mark repeat=4]
      plot coordinates {
        (4,1244/16^3) (5,2132/20^3) (6,9923/36^3) (7,47318/63^3)
        (8,159544/96^3) (9,225344/108^3) (10,1379580/200^3)
        (11,1831620/220^3) (12,4076069/288^3) (13,10052631/390^3)
        (14,19865227/490^3) (15,24417922/525^3)
      };
      \addlegendentry{}
      \addplot[thick,every mark/.append style={solid},
      mark=triangle*,dashed,blue,mark phase=2,mark repeat=4]
      plot coordinates {
        (2,140/4^3) (3,384/7^3) (4,628/10^3) (5,998/13^3) 
        (6,1521/16^3) (7,2224/19^3) (8,3134/22^3) (9,4278/25^3) 
        (10,5683/28^3) (11,7376/31^3) (12,9384/34^3) (13,11734/37^3) 
        (14,14453/40^3) (15,17568/43^3)
      };
      \addlegendentry{}
      \addplot[thick,every mark/.append style={solid},
      mark=triangle*,dashed,green,mark phase=3,mark repeat=4]
      plot coordinates {
        (2,140/4^3) (3,1000/10^3) (4,5350/20^3) (5,22575/35^3) 
        (6,78848/56^3) (7,237160/84^3) (8,633100/120^3) (9,1534225/165^3) 
        (10,3433100/220^3) (11,7186608/286^3) (12,14216930/364^3)
      };
      \addlegendentry{}
      \addplot[thick,every mark/.append style={solid,rotate=180},
      mark=triangle*,dashed,orange,mark phase=4,mark repeat=4]
      plot coordinates {
        (2,31/2^3) (3,57/3^3) (4,76/4^3) (5,101/5^3) 
        (6,133/6^3) (7,173/7^3) (8,222/8^3) (9,281/9^3) 
        (10,351/10^3) (11,433/11^3) (12,528/12^3) (13,637/13^3) 
        (14,761/14^3) (15,901/15^3)
      };
      \addlegendentry{}

      \addlegendimage{empty legend}
      \addlegendentry{\aBMS}
      \addplot[thick,every mark/.append style={solid,rotate=180},
      mark=triangle*,dotted,red,mark phase=1,mark repeat=4]
      plot coordinates {
        (4,25728/16^3) (5,46093/20^3) (6,205772/36^3) (7,919757/63^3)
        (8,2802347/96^3) (9,3880340/108^3) (10,19393287/200^3) 
      };
      \addlegendentry{}
      \addplot[thick,every mark/.append style={solid,rotate=180},
      mark=triangle*,dotted,blue,mark phase=2,mark repeat=4]
      plot coordinates {
        (2,1988/4^3) (3,6477/7^3) (4,16081/10^3) (5,36225/13^3) 
        (6,58536/16^3) (7,104025/19^3) (8,155495/22^3) (9,229920/25^3) 
        (10,317121/28^3) (11,422449/31^3) (12,546251/34^3) (13,710852/37^3) 
        (14,865148/40^3) (15,1106272/43^3)
      };
      \addlegendentry{}
      \addplot[thick,every mark/.append style={solid,rotate=180},
      mark=triangle*,dotted,green,mark phase=3,mark repeat=4]
      plot coordinates {
        (2,1774/4^3) (3,14591/10^3) (4,75757/20^3) (5,298521/35^3) 
        (6,964815/56^3) (7,2689885/84^3) (8,6679544/120^3) (9,15125328/165^3) 
        (10,31763926/220^3) (11,62657181/286^3) (12,117227645/364^3)
      };
      \addlegendentry{}
      \addplot[thick,every mark/.append style={solid,rotate=180},
      mark=triangle*,dotted,orange,mark phase=4,mark repeat=4]
      plot coordinates {
        (2,292/2^3) (3,668/3^3) (4,2213/4^3) (5,3551/5^3) 
        (6,5006/6^3) (7,7142/7^3) (8,9096/8^3) (9,12903/9^3) 
        (10,16536/10^3) (11,20235/11^3) (12,24496/12^3) (13,32110/13^3) 
        (14,39217/14^3) (15,45012/15^3)
      };
      \addlegendentry{}
    \end{axis}
  \end{tikzpicture}
  \caption{Number of basic operations (\textsc{3D}): \asFGLM \& \aBMS}
  \label{fig:basicop3Dadapt}
\end{figure}
It seems that the
\asFGLM always perform fewer operations than the \aBMS
algorithm. Though, it is possible that, in dimension $2$,
for larger parameters,
the \aBMS becomes more efficient than the \asFGLM algorithm as
suggested by the graphs. Concerning the \textsc{L} shape family, although the
\aBMS algorithm do not reduce much its number of table queries, it
performs in fact much fewer basic operations than the \BMS algorithm.
For instance, in~\cite[Section~6]{part1}, we can see that the
\BMS algorithm performs four times (\resp seven times) as many basic
operations as the \aBMS algorithm in dimension $2$ (\resp dimension $3$).

It is also possible that the larger number of operations the \aBMS
algorithm performs compared to the \sFGLM algorithm is due
to the larger number of
queries it needs to recover the relations.

Therefore, we now also compare the ratio
between their number of
basic operations and their number of queries in
Figures~\ref{fig:basicop/queries2Dadapt} and~\ref{fig:basicop/queries3Dadapt}.
\begin{figure}[htbp!]
  \pgfplotsset{
    small,
    width=12cm,
    height=7cm,
    legend cell align=left,
    legend columns=5,
    legend style={at={(-0.1,0.98)},anchor=south west,font=\scriptsize,
    }
  }
  \centering
  \begin{tikzpicture}[baseline]
    \begin{axis}[
      ymode=log,
      xlabel={$d$},
      xlabel style={at={(0.95,0.1)}},
      xmin=3.8,xmax=25.2,
      ymin=3.8,ymax=25000,
      xtick={3,...,25},
      ytick={1,2,3,4,5,6,7,8,9,10,20,30,40,50,60,70,80,90,100,
        200,300,400,500,600,700,800,900,1000,
        2000,3000,4000,5000,6000,7000,8000,9000,10000,
        20000,30000,40000,50000,60000,70000,80000,90000,100000,
        200000},
      yticklabels={},
      extra y ticks={5,10,50,100,500,1000,5000,10000,50000,100000},
      extra y tick labels={5,10,50,100,500,1000,5000,10000,50000,100000},
      ylabel={\#\,Basic Op/\#\,Queries},
      ylabel style={at={(0.08,0.72)}},
      ]
      \addlegendimage{empty legend}
      \addlegendentry{}
      \addlegendimage{legend image with text=Rectangle}
      \addlegendentry{}
      \addlegendimage{legend image with text=\textsc{L} shape}
      \addlegendentry{}
      \addlegendimage{legend image with text=Simplex}
      \addlegendentry{}
      \addlegendimage{legend image with text=Shape Position}
      \addlegendentry{}

      \addlegendimage{empty legend}
      \addlegendentry{\asFGLM}
      \addplot[thick,every mark/.append style={solid},
      mark=triangle*,dashed,red,mark phase=1,mark repeat=4]
      plot coordinates {
        (4,210/25) (5,337/31) (6,1595/60) (7,2355/70) (8,7147/110)
        (9,9871/124) (10,24599/176) (11,32254/194) (12,69587/258)
        (13,87739/280) (14,170015/356) (15,208053/382) (16,371335/470)
        (17,443935/500) (18,742787/600) (19,871616/634)
        (20,1384599/746) (21,1600259/784) (22,2436147/908) (23,2780359/950) 
        (24,4085075/1086) (25,4613103/1132)
        
      };
      \addlegendentry{}
      \addplot[thick,every mark/.append style={solid},
      mark=triangle*,dashed,blue,mark phase=2,mark repeat=4]
      plot coordinates {
        (2,52/10) (3,128/19) (4,215/30) (5,342/43) 
        (6,517/58) (7,748/75) (8,1043/94) (9,1410/115) 
        (10,1857/138) (11,2392/163) (12,3023/190) (13,3758/219) 
        (14,4605/250) (15,5572/283) (16,6667/318) (17,7898/355) 
        (18,9273/394) (19,10800/435) (20,12487/478) (21,14342/523) 
        (22,16373/570) (23,18588/619) (24,20995/670) (25,23602/723)
      };
      \addlegendentry{}
      \addplot[thick,every mark/.append style={solid},
      mark=triangle*,dashed,green,mark phase=3,mark repeat=4]
      plot coordinates {
        (2,52/10) (3,200/21) (4,615/36) (5,1610/55) 
        (6,3724/78) (7,7812/105) (8,15150/136) (9,27555/171) 
        (10,47520/210) (11,78364/253) (12,124397/300) (13,191100/351) 
        (14,285320/406) (15,415480/465) (16,591804/528) (17,826557/595) 
        (18,1134300/666) (19,1532160/741) (20,2040115/820) (21,2681294/903) 
        (22,3482292/990) (23,4473500/1081) (24,5689450/1176) (25,7169175/1275)
      };
      \addlegendentry{}
      \addplot[thick,every mark/.append style={solid},
      mark=triangle*,dashed,orange,mark phase=4,mark repeat=4]
      plot coordinates {
        (2,22/6) (3,41/9) (4,59/11) (5,83/13) 
        (6,114/15) (7,153/17) (8,201/19) (9,259/21) 
        (10,328/23) (11,409/25) (12,503/27) (13,611/29) 
        (14,734/31) (15,873/33) (16,1029/35) (17,1203/37) 
        (18,1396/39) (19,1609/41) (20,1843/43) (21,2099/45) 
        (22,2378/47) (23,2681/49) (24,3009/51) (25,3363/53)
      };
      \addlegendentry{}

      \addlegendimage{empty legend}
      \addlegendentry{\aBMS}
      \addplot[thick,every mark/.append style={solid,rotate=180},
      mark=triangle*,dotted,red,mark phase=1,mark repeat=4]
      plot coordinates {
        (4,2115/41) (5,3108/58) (6,12589/103) (7,18206/132)
        (8,50410/198)
        (9,65117/236) (10,142499/320) (11,180157/371) (12,352288/478)
        (13,426269/541) (14,738578/663) (15,865913/731)
        (16,1432507/882)
        (17,1664165/967) (18,2567128/1141) (19,2926327/1238) (20,4307996/1418)
      };
      \addlegendentry{}
      \addplot[thick,every mark/.append style={solid,rotate=180},
      mark=triangle*,dotted,blue,mark phase=2,mark repeat=4]
      plot coordinates {
        (2,438/10) (3,1111/21) (4,2259/36) (5,3721/55) 
        (6,5817/78) (7,8417/105) (8,11375/136) (9,14881/171) 
        (10,19457/210) (11,23359/253) (12,29387/300) (13,34697/351) 
        (14,40977/406) (15,48195/465) (16,55495/528) (17,63243/595) 
        (18,72697/666) (19,80719/741) (20,91271/820) (21,101213/903) 
        (22,113021/990) (23,122455/1081) (24,136731/1176) (25,148925/1275)
      };
      \addlegendentry{}
      \addplot[thick,every mark/.append style={solid,rotate=180},
      mark=triangle*,dotted,green,mark phase=3,mark repeat=4]
      plot coordinates {
        (2,427/10) (3,1759/21) (4,5241/36) (5,12860/55) 
        (6,27552/78) (7,53414/105) (8,95823/136) (9,161690/171) 
        (10,259672/210) (11,400330/253) (12,596325/300) (13,862630/351) 
        (14,1216684/406) (15,1678622/465) (16,2271453/528) (17,3021248/595) 
        (18,3957182/666) (19,5112366/741) (20,6522843/820) (21,8228829/903) 
        (22,10275367/990) (23,12709953/1081) (24,15585915/1176) (25,18960680/1275)
      };
      \addlegendentry{}
      \addplot[thick,every mark/.append style={solid,rotate=180},
      mark=triangle*,dotted,orange,mark phase=4,mark repeat=4]
      plot coordinates {
        (2,136/6) (3,337/10) (4,686/14) (5,1108/18) 
        (6,1966/22) (7,2794/26) (8,3205/30) (9,4506/34) 
        (10,5598/38) (11,6283/42) (12,8114/46) (13,9678/50) 
        (14,10331/54) (15,12964/58) (16,15034/62) (17,16314/66) 
        (18,18661/70) (19,21618/74) (20,23282/78) (21,26869/82) 
        (22,29037/86) (23,31586/90) (24,35273/94) (25,39316/98)
      };
      \addlegendentry{}
    \end{axis}
  \end{tikzpicture}
  \caption{Number of basic operations by queries (\textsc{2D}):
    \asFGLM \& \aBMS}
  \label{fig:basicop/queries2Dadapt}
\end{figure}

In dimension $2$, the \asFGLM algorithm seems to have a
better ratio between the number of operations and the number of
queries than the \aBMS algorithm. Yet, once again, it is possible that this
statement is not true for larger $d$.

\begin{figure}[htbp!]
  \pgfplotsset{
    small,
    width=12cm,
    height=7cm,
    legend cell align=left,
    legend columns=5,
    legend style={at={(-0.1,0.98)},anchor=south west,font=\scriptsize,
    }
  }
  \centering
  \begin{tikzpicture}[baseline]
    \begin{axis}[
      ymode=log,
      xlabel={$d$},
      xlabel style={at={(0.95,0.1)}},
      xmin=3.7,xmax=15.2,
      ymin=3.8,ymax=25000,
      xtick={2,...,15},
      ytick={1,2,3,4,5,6,7,8,9,10,20,30,40,50,60,70,80,90,100,
        200,300,400,500,600,700,800,900,1000,2000,3000,4000,5000,6000,
        7000,8000,9000,10000,20000,30000,40000,50000,60000},
      yticklabels={},
      extra y ticks={5,10,50,100,500,1000,5000,10000,50000},
      extra y tick labels={5,10,50,100,500,1000,5000,10000,50000},
      ylabel={\#\,Basic Op/\#\,Queries},
      ylabel style={at={(0.08,0.72)}},
      ]
      \addlegendimage{empty legend}
      \addlegendentry{}
      \addlegendimage{legend image with text=Rectangle}
      \addlegendentry{}
      \addlegendimage{legend image with text=\textsc{L} shape}
      \addlegendentry{}
      \addlegendimage{legend image with text=Simplex}
      \addlegendentry{}
      \addlegendimage{legend image with text=Shape Position}
      \addlegendentry{}

      \addlegendimage{empty legend}
      \addlegendentry{\asFGLM}
      \addplot[thick,every mark/.append style={solid},
      mark=triangle*,dashed,red,mark phase=1,mark repeat=4]
      plot coordinates {
        (4,1244/72) (5,2132/90) (6,9923/174) (7,47318/334)
        (8,159544/534) (9,225344/604) (10,1379580/1206) (11,1831620/1332) 
        (12,4076069/1780) (13,10052631/2484) (14,19865227/3168) (15,24417922/3402)
      };
      \addlegendentry{}
      \addplot[thick,every mark/.append style={solid},
      mark=triangle*,dashed,blue,mark phase=2,mark repeat=4]
      plot coordinates {
        (2,140/20) (3,384/42) (4,628/69) (5,998/102) 
        (6,1521/141) (7,2224/186) (8,3134/237) (9,4278/294) 
        (10,5683/357) (11,7376/426) (12,9384/501) (13,11734/582) 
        (14,14453/669) (15,17568/762) 
      };
      \addlegendentry{}
      \addplot[thick,every mark/.append style={solid},
      mark=triangle*,dashed,green,mark phase=3,mark repeat=4]
      plot coordinates {
        (2,140/20) (3,1000/56) (4,5350/120) (5,22575/220) 
        (6,78848/364) (7,237160/560) (8,633100/816) (9,1534225/1140) 
        (10,3433100/1540) (11,7186608/2024) (12,14216930/2600)
      };
      \addlegendentry{}
      \addplot[thick,every mark/.append style={solid},
      mark=triangle*,dashed,orange,mark phase=4,mark repeat=4]
      plot coordinates {
        (2,31/8) (3,57/13) (4,76/18) (5,101/21) 
        (6,133/25) (7,173/30) (8,222/33) (9,281/36) 
        (10,351/40) (11,433/46) (12,528/49) (13,637/52) 
        (14,761/55) (15,901/59)
      };
      \addlegendentry{}

      \addlegendimage{empty legend}
      \addlegendentry{\aBMS}
      \addplot[thick,every mark/.append style={solid,rotate=180},
      mark=triangle*,dotted,red,mark phase=1,mark repeat=4]
      plot coordinates {
        (4,25728/207) (5,46093/330) (6,205772/709) (7,919757/1265)
        (8,2802347/2159) (9,3880340/2739) (10,19393287/4819) 
      };
      \addlegendentry{}
      \addplot[thick,every mark/.append style={solid,rotate=180},
      mark=triangle*,dotted,blue,mark phase=2,mark repeat=4]
      plot coordinates {
        (2,1988/20) (3,6477/56) (4,16081/120) (5,36225/218) 
        (6,58536/357) (7,104025/545) (8,155495/784) (9,229920/1090) 
        (10,317121/1457) (11,422449/1907) (12,546251/2424) (13,710852/3043) 
        (14,865148/3741) (15,1106272/4557)
      };
      \addlegendentry{}
      \addplot[thick,every mark/.append style={solid,rotate=180},
      mark=triangle*,dotted,green,mark phase=3,mark repeat=4]
      plot coordinates {
        (2,1774/20) (3,14591/56) (4,75757/120) (5,298521/220) 
        (6,964815/364) (7,2689885/560) (8,6679544/816) (9,15125328/1140) 
        (10,31763926/1540) (11,62657181/2024) (12,117227645/2600) 
      };
      \addlegendentry{}
      \addplot[thick,every mark/.append style={solid,rotate=180},
      mark=triangle*,dotted,orange,mark phase=4,mark repeat=4]
      plot coordinates {
        (2,292/8) (3,668/13) (4,2213/18) (5,3551/23) 
        (6,5006/28) (7,7142/33) (8,9096/38) (9,12903/43) 
        (10,16536/48) (11,20235/53) (12,24496/58) (13,32110/63) 
        (14,39217/68) (15,45012/73)
      };
      \addlegendentry{}
    \end{axis}
  \end{tikzpicture}
  \caption{Number of basic operations by queries (\textsc{3D}):
    \asFGLM \& \aBMS}
  \label{fig:basicop/queries3Dadapt}
\end{figure}

In dimension $3$, however, our experiments lead us to believe that
this ratio will always be larger for the \aBMS algorithm than for the
\asFGLM algorithm.



\bibliographystyle{elsarticle-harv} 
\addcontentsline{toc}{section}{References}
\bibliography{biblio}

\begin{thebibliography}{32}
\expandafter\ifx\csname natexlab\endcsname\relax\def\natexlab#1{#1}\fi
\expandafter\ifx\csname url\endcsname\relax
  \def\url#1{\texttt{#1}}\fi
\expandafter\ifx\csname urlprefix\endcsname\relax\def\urlprefix{URL }\fi

\bibitem[{Banderier and Flajolet(2002)}]{BanderierF2002}
Banderier, C., Flajolet, P., 2002. Basic analytic combinatorics of directed
  lattice paths. Theoret.\ Comput. Sci. 281~(1--2), 37--80, selected Papers in
  honour of Maurice Nivat.
\newline\urlprefix\url{http://www.sciencedirect.com/science/article/pii/S0304397502000075}

\bibitem[{Benoit et~al.(2010)Benoit, Chyzak, Darrasse, Gerhold, Mezzarobba, and
  Salvy}]{DDMF}
Benoit, A., Chyzak, F., Darrasse, A., Gerhold, S., Mezzarobba, M., Salvy, B.,
  2010. {The Dynamic Dictionary of Mathematical Functions (DDMF)}. In: Fukuda,
  K., Hoeven, J. v.~d., Joswig, M., Takayama, N. (Eds.), Mathematical Software
  -- ICMS~2010. Springer, Berlin, Heidelberg, pp. 35--41.
\newline\urlprefix\url{http://dx.doi.org/10.1007/978-3-642-15582-6_7}

\bibitem[{Berlekamp(1968)}]{Berl68}
Berlekamp, E., 1968. Nonbinary {BCH} decoding. IEEE Trans.\ Inform.\ Theory
  14~(2), 242--242.

\bibitem[{Berthomieu et~al.(2015)Berthomieu, Boyer, and
  Faug{\`e}re}]{issac2015}
Berthomieu, J., Boyer, B., Faug{\`e}re, J.-{\relax Ch}., 2015. {Linear Algebra
  for Computing Gr\"obner Bases of Linear Recursive Multidimensional
  Sequences}. In: {40th International Symposium on Symbolic and Algebraic
  Computation}. Proceedings of the 40th International Symposium on Symbolic and
  Algebraic Computation. Bath, United Kingdom, pp. 61--68.

\bibitem[{Berthomieu et~al.(2017)Berthomieu, Boyer, and
  Faug{\`e}re}]{berthomieu:hal-01253934}
Berthomieu, J., Boyer, B., Faug{\`e}re, J.-{\relax Ch}., 2017. {Linear Algebra
  for Computing Gr{\"o}bner Bases of Linear Recursive Multidimensional
  Sequences}. {Journal of Symbolic Computation} 83~(Supplement C), 36--67,
  special issue on the conference ISSAC 2015: Symbolic computation and computer
  algebra.
\newline\urlprefix\url{https://hal.inria.fr/hal-01253934}

\bibitem[{Berthomieu and Faug{\`e}re(2016)}]{berthomieu:hal-01314266}
Berthomieu, J., Faug{\`e}re, J.-{\relax Ch}., 2016. {Guessing Linear Recurrence
  Relations of Sequence Tuples and P-recursive Sequences with Linear Algebra}.
  In: {41st International Symposium on Symbolic and Algebraic Computation}.
  Waterloo, ON, Canada, pp. 95--102.

\bibitem[{Berthomieu and Faug{\`e}re(2017)}]{part1}
Berthomieu, J., Faug{\`e}re, J.-{\relax Ch}., 2017. {In-depth comparison of the
  Berlekamp -- Massey -- Sakata and the Scalar-FGLM algorithms: the non
  adaptive variants}, preprint.
\newline\urlprefix\url{https://hal.inria.fr/hal-01516708}

\bibitem[{Berthomieu and Faug{\`e}re(2018)}]{berthomieu:hal-01784369}
Berthomieu, J., Faug{\`e}re, J.-{\relax Ch}., 2018. {A
  Polynomial-Division-Based Algorithm for Computing Linear Recurrence
  Relations}. In: {ISSAC 2018 - 43rd International Symposium on Symbolic and
  Algebraic Computation}. New York, United States, p.~8.
\newline\urlprefix\url{https://hal.inria.fr/hal-01784369}

\bibitem[{Bose and Ray-Chaudhuri(1960)}]{BoseRC1960}
Bose, R., Ray-Chaudhuri, D., 1960. On a class of error correcting binary group
  codes. Information and Control 3~(1), 68 -- 79.
\newline\urlprefix\url{http://www.sciencedirect.com/science/article/pii/S0019995860902874}

\bibitem[{Bostan et~al.(2014)Bostan, Bousquet-M\'elou, Kauers, and
  Melczer}]{BostanBMKM2014}
Bostan, A., Bousquet-M\'elou, M., Kauers, M., Melczer, S., 2014. On
  3-dimensional lattice walks confined to the positive octant, to appear in
  Annals of Combinatorics.

\bibitem[{Bousquet-M{\'e}lou and Mishna(2010)}]{BousquetMM2010}
Bousquet-M{\'e}lou, M., Mishna, M., 2010. Walks with small steps in the quarter
  plane. In: Algorithmic probability and combinatorics. Vol. 520 of Contemp.
  Math. Amer. Math. Soc., Providence, RI, pp. 1--39.
\newline\urlprefix\url{http://dx.doi.org/10.1090/conm/520/10252}

\bibitem[{Bousquet-M\'{e}lou and Petkov\v{s}ek(2003)}]{BousquetMP2003}
Bousquet-M\'{e}lou, M., Petkov\v{s}ek, M., 2003. Walks confined in a quadrant
  are not always d-finite. Theoret.\ Comput.\ Sci. 307~(2), 257--276, random
  Generation of Combinatorial Objects and Bijective Combinatorics.
\newline\urlprefix\url{http://www.sciencedirect.com/science/article/pii/S0304397503002196}

\bibitem[{Bras-Amor{\'o}s and O'Sullivan(2006)}]{Bras-Amoros2006}
Bras-Amor{\'o}s, M., O'Sullivan, M.~E., 2006. The correction capability of the
  {Berlekamp--Massey--Sakata} algorithm with majority voting. Applicable
  Algebra in Engineering, Communication and Computing 17~(5), 315--335.
\newline\urlprefix\url{http://dx.doi.org/10.1007/s00200-006-0015-8}

\bibitem[{Cox et~al.(2015)Cox, Little, and O'Shea}]{CoxLOS2015}
Cox, D., Little, J., O'Shea, D., 2015. Ideals, {V}arieties, and {A}lgorithms,
  4th Edition. Undergraduate Texts in Mathematics. Springer, New York, an
  introduction to computational algebraic geometry and commutative algebra.

\bibitem[{Cox et~al.(2005)Cox, Little, and O'Shea}]{CoxLOS2015b}
Cox, D.~A., Little, J., O'Shea, D., 2005. Using Algebraic Geometry, 2nd
  Edition. Vol. 185 of Graduate Texts in Mathematics. Springer, New York.

\bibitem[{Faug{\`e}re et~al.(1993)Faug{\`e}re, Gianni, Lazard, and
  Mora}]{FGLM93}
Faug{\`e}re, J.-{\relax Ch}., Gianni, P., Lazard, D., Mora, T., 1993. Efficient
  {C}omputation of {Z}ero-dimensional {G}r{\"o}bner {B}ases by {C}hange of
  {O}rdering. J.\ Symbolic Comput. 16~(4), 329--344.

\bibitem[{Faug{\`e}re and Mou(2011)}]{FM11}
Faug{\`e}re, J.-{\relax Ch}., Mou, C., 2011. {Fast Algorithm for Change of
  Ordering of Zero-dimensional {Gr\"obner} Bases with Sparse Multiplication
  Matrices}. In: {Proc.\ of the 36th ISSAC}. ACM, pp. 115--122.

\bibitem[{Faug{\`e}re and Mou(2017)}]{faugere:hal-00807540}
Faug{\`e}re, J.-{\relax Ch}., Mou, C., 2017. {Sparse FGLM algorithms}. {Journal
  of Symbolic Computation} 80~(3), 538 -- 569.

\bibitem[{Fitzpatrick and Norton(1990)}]{FitzpatrickN90}
Fitzpatrick, P., Norton, G., 1990. Finding a basis for the characteristic ideal
  of an n-dimensional linear recurring sequence. IEEE Trans. Inform. Theory
  36~(6), 1480--1487.

\bibitem[{Guisse(2016)}]{Guisse2016}
Guisse, V., 2016. Alg\`ebre lin\'eaire d\'edi\'ee pour les algorithmes
  {\textsc{scalar-fglm}} et {Berlekamp-Massey-Sakata}. Master's thesis,
  Universit\'e Paris-Diderot.

\bibitem[{Hocquenghem(1959)}]{Hocquenghem1959}
Hocquenghem, A., 1959. Codes correcteurs d'erreurs. Chiffres 2, 147 -- 156.

\bibitem[{Jonckheere and Ma(1989)}]{JoMa89}
Jonckheere, E., Ma, C., 1989. A simple {H}ankel interpretation of the
  {Berlekamp-Massey} algorithm. Linear Algebra Appl. 125~(0), 65 -- 76.
\newline\urlprefix\url{http://www.sciencedirect.com/science/article/pii/0024379589900323}

\bibitem[{Kaltofen and Pan(1991)}]{KaPa91}
Kaltofen, E., Pan, V., 1991. Processor efficient parallel solution of linear
  systems over an abstract field. In: SPAA '91. ACM Press, New York, N.Y., pp.
  180--191.

\bibitem[{Kaltofen and Yuhasz(2013{\natexlab{a}})}]{KaYa13}
Kaltofen, E., Yuhasz, G., 2013{\natexlab{a}}. A fraction free {M}atrix
  {Berlekamp/Massey} algorithm. Linear Algebra Appl. 439~(9), 2515--2526.

\bibitem[{Kaltofen and Yuhasz(2013{\natexlab{b}})}]{Kalto06}
Kaltofen, E., Yuhasz, G., 2013{\natexlab{b}}. On the {M}atrix
  {Berlekamp-Massey} {A}lgorithm. ACM Trans.\ Algorithms 9~(4), 33:1--33:24.
\newline\urlprefix\url{http://doi.acm.org/10.1145/2500122}

\bibitem[{Levinson(1947)}]{Levinson47}
Levinson, N., 1947. The {Wiener RMS} {(Root-Mean-Square)} error criterion in
  the filter design and prediction. J.\ Math.\ Phys. 25, 261--278.

\bibitem[{Massey(1969)}]{Mass69}
Massey, J.~L., 1969. Shift-register synthesis and {BCH} decoding. IEEE Trans.\
  Inform.\ Theory {\sc it}-15, 122--127.

\bibitem[{Sakata(1988)}]{Sakata88}
Sakata, S., 1988. Finding a minimal set of linear recurring relations capable
  of generating a given finite two-dimensional array. J.\ Symbolic Comput.
  5~(3), 321--337.
\newline\urlprefix\url{http://www.sciencedirect.com/science/article/pii/S0747717188800336}

\bibitem[{Sakata(1990)}]{Sakata90}
Sakata, S., 1990. Extension of the {B}erlekamp-{M}assey algorithm to {$N$}
  {D}imensions. Inform.\ and Comput. 84~(2), 207--239.
\newline\urlprefix\url{http://dx.doi.org/10.1016/0890-5401(90)90039-K}

\bibitem[{Sakata(1991)}]{Sakata91}
Sakata, S., 1991. Decoding binary {$2$-D} cyclic codes by the {$2$-D}
  {Berlekamp-Massey} algorithm. {IEEE} Trans.\ Inform.\ Theory 37~(4),
  1200--1203.
\newline\urlprefix\url{http://dx.doi.org/10.1109/18.86974}

\bibitem[{Sakata(2009)}]{Sakata09}
Sakata, S., 2009. The {BMS A}lgorithm. In: Sala, M., Sakata, S., Mora, T.,
  Traverso, C., Perret, L. (Eds.), Gr{\"o}bner Bases, Coding, and Cryptography.
  Springer Berlin Heidelberg, Berlin, Heidelberg, pp. 143--163.
\newline\urlprefix\url{http://dx.doi.org/10.1007/978-3-540-93806-4_9}

\bibitem[{Wiener(1964)}]{Wiener49}
Wiener, N., 1964. Extrapolation, {I}nterpolation, and {S}moothing of
  {S}tationary {T}ime {S}eries. The MIT Press.

\end{thebibliography}

\begin{appendices}
  \section{The \BMS algorithm}\label{a:BMS}
  \renewcommand{\thetheorem}{\thesection.\arabic{theorem}}
\renewcommand{\thedefinition}{\thesection.\arabic{definition}}
\renewcommand{\theexample}{\thesection.\arabic{example}}
\renewcommand{\thealgo}{\thesection.\arabic{algo}}
\renewcommand{\thealgorithm}{\thesection.\arabic{algorithm}}
\setcounter{theorem}{0}
\setcounter{definition}{0}
\setcounter{example}{0}
\setcounter{algo}{0}
\setcounter{algorithm}{0}
\emph{This appendix can also be found in~\cite[Section~3]{part1}.}

As in~\cite{Guisse2016}, we specialize to $\K[\x]$ the presentation of
the \BMS algorithm
given in~\cite{Bras-Amoros2006}, \cite{CoxLOS2015b} and~\cite{Sakata09}
in the more general case of ordered domains.

\subsection{A Polynomial interpretation of the \BMS algorithm}
Given a table $\bu=(u_{\bi})_{\bi\in\N^n}$ and a weight
ordering $\prec$ for $\x$. We let
$\cT_0=\{0\}\cup\{\x^{\bi},\ \bi\in\N^n\}$
and extend $\prec$ (still denoted by $\prec$)
to $\cT_0$ with the convention that $0\prec 1$.

The goal is to iterate on a monomial $m$,
by only considering, at each step, the table
$(u_{\bi})_{\bi\in\{\bk,\ \x^{\bk}\preceq m\}}$. As we only know partially the
table $\bu$, we need to define some notions according to this partial
knowledge at step $m$.
\begin{definition}
  \label{def:shift}
  Let $m\in\cT_0$. Let $f\in\K[\x]$, we say that the relation $f$ is
  \emph{valid up to $m$}, whenever
  \[\forall t\in\cT_0,\,
  \LM(t\,f) \preceq m \Rightarrow [t\,f]=0.\]
  We thus define the \emph{shift} of
  $f$ as $\shift (f)=\frac{m}{\LM(f)}$.

  We say that the relation $f$ \emph{fails} at $m$ whenever
  \begin{align*}
    \forall t\in\cT_0,\,
    t\,f \prec m \Rightarrow [t\,f]&=0,\\
    \cro{\frac{m}{\LM(f)}\,f}&\neq 0.
  \end{align*}
  We define the \emph{fail} of $f$ as
  $\fail(f)=m$.
  If the relation $f$ never fails, that is for all $t\in\cT_0$, $[t\,f]=0$,
  then by convention $\fail(f)=\shift(f)=+\infty$.
\end{definition}

\begin{proposition}
  Let $\bu$ be a table and $f\in\K[x]$ such that $\fail(f)\succ m$. For all
  $g\in\K[\x]$, if $\LM(g\,f)\preceq m$, then $[g\,f]=0$.
\end{proposition}
The following proposition show how to combine two failing relations with the
same shift in order to obtain a new relation valid with a bigger shift.

\begin{proposition}
  \label{prop:augm_shift}
  Let $f_1$ and $f_2$ be two relations such that
  $v=\frac{\fail (f_1)}{\LM(f_1)}=\frac{\fail (f_2)}{\LM (f_2)}$ and
  $e_1=\cro{v\,f_1}$, $e_2=\cro{v\,f_2}$. Let $f$ be the nonzero polynomial
  $f_1 - \frac{e_1}{e_2}\,f_2$. Then, for $i\in\{1,2\}$,
  $\fail(f)\succ \fail(f_i)$, \ie
  $\frac{\fail(f)}{\LM(f)}\succ v$.
\end{proposition} 
\begin{proof}
  For any $c\in\K$ and any $\mu\in\K[\x]$ such that
  $\LM(g)\prec v$, 
  we have
  $[\mu\,(f_1+c\,f_2)]=[\mu\,f_1]+c\,[\mu\,f_2]=0$, hence
  $\fail (f_1+c\,f_2)\succeq \fail (f_i)$.
  
  It remains to prove that for a good choice of $c$, we have a strict
  inequality:
  as, $[v\,(f_1+c\,f_2)]=[v\,f_1]+c\,[v\,f_2]=e_1+c\,e_2$, it is clear that
  $[v\,f]=[v\,(f_1-\frac{e_1}{e_2}\,f_2)]=0$, so that
  $\fail (f)\succ v\,\LM(f)\succeq\fail (f_i)$.
\end{proof}

\begin{definition}\label{def:im}
  Using the same notation as in Definition~\ref{def:staircase}, we let
  \[I_m = \{f\in\K[\x],\ \fail\pare{f}\succ m\},\]
  and $\cG_m$ be the least elements for $\prec$ of $I_m$, it is
  a truncated \gb of $I_m$:
  \begin{align*}
    \cG_m &= \min_{\prec}\{g,\ g\in I_m\},\\
    S_m &= \Staircase(\cG_m).
  \end{align*}
\end{definition}

\begin{example}
  Let us go back to Example~\ref{ex:binom} with sequence
  $\bin = \left(\binom{i}{j}\right)_{(i,j)\in\N^2}$.
  Consider $\K[x,y]$ with the $\DRL(y\prec x)$ ordering,
  and $m=x^2$.
  \[
  \begin{ytableau}
    \none[y^2] & 0   \\
    \none[y]   & 0        & 1   \\
    \none[1]   & 1        & 1        &     *(green)1 \\  
    \none      & \none[1] & \none[x] & \none[x^2]
  \end{ytableau}
  \]
  From this table, on the one hand, we can deduce that
  \begin{itemize}
  \item since it is not identically $0$, there is no relation with
    leading monomial $1$ valid up to $x^2$, hence $1\in S_{x^2}$;
  \item since $[y+\alpha]=\alpha$ and
    $[x\,(y+\alpha)]=1+\alpha$, there is no relation with leading
    monomial $y$ valid up to $x\,y$ and thus $x^2$, hence
    $y\in S_{x^2}$;
  \item since $[y\,(x+\beta\,y+\alpha)]=1$, there is no relation with leading
    monomial $x$ valid up to $x\,y$ and thus $x^2$, hence $x\in S_{x^2}$.
  \end{itemize}
  On the other hand, we can check that
  \begin{itemize}
  \item since $[y^2]=0$, relation $y^2$ is valid up to $y^2$ and thus
    $x^2$, hence $y^2\in\cT\setminus S_{x^2}$;
  \item since $[x\,y-1]=0$, relation $x\,y-1$ is valid up to $x\,y$
    and thus $x^2$, hence $x\,y\in\cT\setminus S_{x^2}$;
  \item since $[x^2-x]=0$, relation $x^2-x$ is valid up to $x^2$,
    hence $x^2\in\cT\setminus S_{x^2}$.
  \end{itemize}
  Therefore, $S_{x^2} = \{1,y,x\}$,
  $\max_|(S_{x^2})=\{y,x\}$ and $\min_|(\cT\setminus S_{x^2})=\{y^2,x\,y,x^2\}$.
  This is summed up in the following diagram. 
  
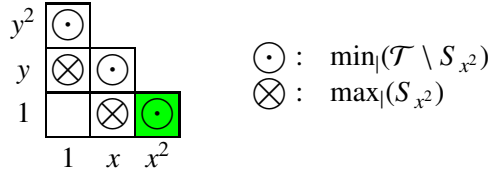
\[
  \begin{ytableau}
    \none[y^2] & \bigodot  \\
    \none[y]   & \bigotimes  & \bigodot\\
    \none[1]   &    & \bigotimes     &    *(green)\bigodot \\
    \none      & \none[1] & \none[x] & \none[x^2]
  \end{ytableau}
  \quad\quad
  \begin{array}{cl}
    \\\\
    \bigodot: &\min_|(\cT\setminus S_{x^2})\\
    \bigotimes: &\max_|(S_{x^2})
  \end{array}
  \]
  Let us notice that many relations with respective leading monomials
  $y^2,x\,y,x^2$ suit actually. These would be
  $y^2-\alpha_1\,x+\alpha_y\,y+\alpha_1,x\,y-(1+\alpha_1)\,x
  +\alpha_y\,y+\alpha_1$ and
  $x^2-(1+\alpha_1)\,x+\alpha_y\,y+\alpha_1$. Furthermore,
  $I_{x^2}$ is not stable by addition: $(x^2-x),(x^2-2\,x+1)\in I_{x^2}$ but
  $x^2-x-(x^2-2\,x+1)=(x-1)\not\in I_{x^2}$ since $\fail\pare{x-1}=x\,y$.
  Hence, $I_{x^2}$ is not an ideal of $\K[x,y]$.

  For $m=x^3$, with the following table, we find that
  
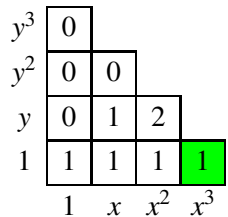
\[
  \begin{ytableau}
    \none[y^3] &0\\
    \none[y^2] & 0 &0  \\
    \none[y]   & 0        & 1   & 2\\
    \none[1]   & 1        & 1        & 1 &     *(green)1 \\  
    \none      & \none[1] & \none[x] & \none[x^2] &\none[x^3]
  \end{ytableau}
  \]
  \begin{itemize}
  \item since $[y^2]=[y\,y^2]=[x\,y^2]=0$, then $y^2$ is valid up to
    $x\,y^2$ and thus $x^3$;
  \item since $[x\,y-1]=[y\,(x\,y-1)]=0$ and $[x\,(x\,y-y)]=1$, then
    $x\,y-1$ fails at $x^2\,y$. Yet, since
    $[y]=[y\,y]=0$ and $[x\,y]=1$, then by
    Proposition~\ref{prop:augm_shift}, 
    $[x\,y-y-1]=[y\,(x\,y-y-1)]=0$ and $[x\,(x\,y-y-1)]$ vanishes as
    well. Hence, $x\,y-y-1$ is
    valid up to $x^2\,y$ and thus $x^3$;
  \item since $[x^2-x]=0$ and $[y\,(x^2-x)]=1$, then $x^2-x$ fails at
    $x^2\,y$. Likewise, since $[x-1]=0$ and $[y\,(x-1)]=1$, then
    $[x^2-2\,x+1]=0$ and $[y\,(x^2-2\,x+1)]=0$. Furthermore,
    $[x\,(x^2-2\,x+1)]=0$, so that
    $x^2-2\,x+1$ is valid up to $x^3$.
  \end{itemize}
  Therefore, $S_{x^3} = \{1,y,x\}$,
  $\max_|(S_{x^3})=\{y,x\}$ and $\min_|(\cT\setminus
  S_{x^3})=\{y^2,x\,y,x^2\}$.
  We can also check that these relations span the only valid relations with
  support in $S_{x^3}\cup\{y^2,x\,y,x^2\}$.
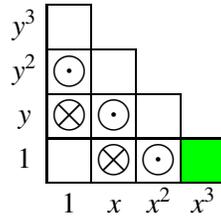
  \[
  \begin{ytableau}
    \none[y^3] &       \\
    \none[y^2] & \bigodot &  \\
    \none[y]   & \bigotimes  & \bigodot & \\
    \none[1]   &    & \bigotimes     &    \bigodot  & *(green)\\
    \none      & \none[1] & \none[x] & \none[x^2] & \none[x^3]\\
  \end{ytableau}
  \]
\end{example}

Although $I_m$ is not an ideal in general, we have the following results:

 \begin{proposition}\label{prop:Im-closed}
   Using the notation of Definitions~\ref{def:shift} and~\ref{def:im},
   \begin{enumerate}
   \item \label{eq:stab} $I_m$ is closed under multiplication by 
   elements of $\K[\x]$,
   \item for all monomials $t,t'$ such that $t| t'$,
     \begin{enumerate}
     \item \label{it:staircase}
       if $t'\in S_m$, then $t\in S_m$.
     \item \label{it:compl_staircase}
       if $t\in\cT\setminus S_m$, then $t'\in\cT\setminus S_m$,
     \end{enumerate}
   \end{enumerate}
 \end{proposition}

Moreover, it is clear that the sequence $(I_m)_{m\in\cT_0}$ is decreasing 
and that if $\bu$ is linear recurrent then $I = \bigcap_{m\in\cT_0}I_m$.
Therefore, $\left(S_m \right)_{m\in\cT_0}$
is increasing and its limit is $S$ the finite target staircase. Hence,
for $m$ big enough, $S_m$
will be the target staircase. We will give an upper bound in 
Proposition~\ref{prop:upperbound}.

The following result gives an intrinsic characterization of 
$S_m$ that is key in the iteration of the \BMS algorithm.

\begin{proposition}\label{prop:iter}
  For all monomial $m\in\cT_0$, 
  $S_m = \acc{\frac{\fail(f)}{\LM(f)},\ f\notin I_m}$.

  Furthermore, let $m^+$ be the successor of $m$. Let $s$ be a
  monomial in the staircase $S_{m^+}$. Then, $s$
  was added at step $m^+$, \ie $s\notin S_m$, if, and only if,
  $s|m^+$ and $\frac{m^+}{s}\in S_{m^+}\setminus S_m$.
\end{proposition}
\begin{proof}
  We shall prove the first assertion by double inclusion.
  If $s=\frac{\fail(f)}{\LM(f)}$ then for all $g\in\K[\x]$ such that
  $\LM(g)=s$, $\fail(g)\preceq m$, hence $s\notin\LM(I_m)$, $s\in S_m$.

  The reverse inclusion is proved by induction on $m$. For $m=0$,
  $S_m=\emptyset$ and there is nothing to do. Let us
  assume the inclusion
  is satisfied for a monomial $m$.
  
  Let $s\in S_{m^+}$. On the one hand, if $s\in S_m$, then
  there exists
  $f\in\K[\x]\setminus I_m\subseteq\K[\x]\setminus I_{m^+}$
  such that $s=\frac{\fail(f)}{\LM(f)}$.

  If, on the other hand, $s\in S_{m^+}\setminus S_m$, then there
  exists a relation $f\in\K[\x]$ such that $\LM(f)=s$,
  and $m\prec \fail(f) \preceq m^+$, hence $\fail(f)=m^+$ and $s$
  divides $m^+$.

  Let us assume that for all $g\in\K[\x]$ with $\LM(g)=\frac{m^+}{s}$, we have
  $\fail (g)\preceq m\prec m^+$. Therefore, $\frac{m^+}{s}\in S_m$ and
  there exists $h\notin I_m$ such that
  $\frac{\fail(h)}{\LM(h)}=\frac{m^+}{s}$. By
  Proposition~\ref{prop:augm_shift}, there is $\alpha\in\K$ such that
  $\fail(f-\alpha\,h)\succ m^+$. Since $\fail(h)\preceq m\prec
  m^+$, then $\LM(h)\preceq s$ and $\LM(f-\alpha\,h)=s$, hence
  $\frac{\fail(f-\alpha\,h)}{\LM(f-\alpha\,h)}\succ\frac{m^+}{s}$. This
  contradicts the fact that $\frac{m^+}{s}\in S_m$. Thus there exists
  a $g\in\K[\x]$ with $\LM(g)=\frac{m^+}{s}$ and $\fail(g)\succeq m^+$.
 
  Let $g$ be such a relation,
  since $\fail(f)=m^+$, then $[g\,f]\neq0$ and
  $\fail(g)=m^+$. Therefore,
  $\frac{\fail(g)}{\LM(g)}=\frac{m^+}{m^+/s}=s$
  so that $s\in\acc{\frac{\fail(f)}{\LM(f)},\ f\notin I_{m^+}}$.

  Now, we proved that $s\in S_{m^+}\setminus S_m$ implies $s|m^+$ and 
  $\frac{m^+}{s}\in S_{m^+}\setminus S_m$. This implication is clearly
  an equivalence.
\end{proof}

From this proposition it follows that
if $m\in\cT_0$, and if $m^+$
is its successor:
\begin{equation}\label{eq:recdelta}
  \max_|(S_{m^+}) = \max_{|}\pare{\max_|(S_m) \cup 
    \acc{\frac{m^+}{s},\ s\in\min_|(\cT\setminus S_m)\cap S_{m^+}}}
\end{equation}

Relation~\ref{eq:recdelta} allows us
to construct, iterating on the monomial $m$, 
the set of relations $G_m$ representing the truncated \gb of $I_m$.
Relations $g\in G_m$ are indexed by their leading
monomials, describing $\cT\setminus S_m$. 

\begin{remark}\label{rk:choose_Sm}
  We can also construct another set, describing the edge of $S_m$,
  still denoted $S_m$, as there is a one-to-one correspondence between a
  staircase and its edge. The relations $h\in S_m$ are indexed by their ratio
  $\frac{\fail(h)}{\LM(h)}$ between their fail and their leading
  monomial, describing the full staircase of $I_m$.

  When two relations $h$ and $h'$ in $S_m$ are such that
  $\frac{\fail(h)}{\LM(h)}=\frac{\fail(h')}{\LM(h')}$, then we only
  need to keep one. Since the goal is to combine a relation of
  $S_m$ with a relation failing at $m^+$ to make a new one with a bigger
  shift, as in Proposition~\ref{prop:augm_shift}, it is best to handle
  smaller polynomials.
\end{remark}
This yields Algorithm~\ref{algo:bms}.

\begin{algorithm2e}[htbp!]\label{algo:bms}
  \small
  \DontPrintSemicolon
  \TitleOfAlgo{The \BMS algorithm.}
  \KwIn{A table $\bu=(u_{\bi})_{\bi\in\N^n}$ with coefficients in
    $\K$, a monomial ordering $\prec$ and a monomial $M$
    as the stopping condition.}
  \KwOut{A set $G$ of relations generating $I_M$.}
  $T := \{m\in\K[\x],\ m\preceq M\}$.\tcp*{ordered for $\prec$}
  $G := \{1\}$.\tcp*{the future \gb}
  $S := \emptyset$.\tcp*{staircase edge, elements will be
    $[h,\fail(h)/\LM(h)]$} 
  
  \Forall{$m\in T$}{
    $S' := S$.\;
    \For{$g\in G$}{
      \If{$\LM(g)| m$}{
        $e:=\cro{\frac{m}{\LM(g)}\,g}_{\bu}$.\;
        \If{$e\neq 0$}{
          $S':=S'\cup\acc{\cro{\frac{g}{e},\frac{m}{\LM(g)}}}$.\;
        }
      }
    }
    $S':=\min_|\acc{[h,\fail(h)/\LM(h)]}$.
    \tcp*{see Remark~\ref{rk:choose_Sm}}
    $G':= \Border (S')$.\;
    
    \For{$g'\in G'$}{
      Let $g\in G$ such that $\LM(g)|\LM(g')$.\; 
      \uIf{$\LM(g) \nmid m$}{
        $g':=\frac{\LM(g')}{\LM(g)}\,g$.\tcp*{translates the relation}
      }
      \uElseIf{$\exists\,h\in S,
        \frac{m}{\LM(g')} |\fail(h)$}{
        $g':= \frac{\LM(g')}{\LM(g)}\,g
        -\cro{\frac{m}{\LM(g)}\,g}_{\bu}\,
        \frac{\LM(g')\,\fail(h)}{m}\,h$.\tcp*{see
          Proposition~\ref{prop:augm_shift}}
      }
      \lElse{
        $g':=g$.
      }
    }
    $G := G'$.\;
    $S := S'$.\;
  }  
  \KwRet $G$.
\end{algorithm2e}

We saw that for $m$ big enough, $S_m$
will be the target staircase. We now
give an upper bound.

\begin{proposition}\label{prop:upperbound}
  Let $\bu$ be a linear recurrent sequence and $I$ be its ideal of
  relations.

  Let $S$ be the
  staircase of $I$ for $\prec$. Let $s_{\max}$
  be the largest monomial in
  $S$. Then, for $m\succeq (s_{\max})^2$,
  $S_m = S$.

  Let $\cG$ be a minimal \gb of $I$ for $\prec$ and let $g_{\max}$ be the
  largest leading monomial of  $\cG$. Then, for $m\succeq
  s_{\max}\cdot\max_{\prec}(g_{\max},s_{\max})$,
  the \BMS algorithm returns a minimal \gb of $I$
  for $\prec$.
\end{proposition}
\begin{example}
  For the $\DRL(y\prec x)$ ordering, $I=\langle x^p,y^q\rangle$ and
  $q>p\geq 1$, we have, $s_{\max}=x^{p-1}\,y^{q-1}$ and
  $g_{\max}=y^q$. Therefore, the right staircase is found at most at step
  $m=x^{2\,p-2}\,y^{2\,q-2}$, while the \gb is found at most at step
  $x^{p-1}\,y^{q-1}\,\max_{\prec}(x^{p-1}\,y^{q-1},y^q)$, \ie
  $y^{2\,q-1}$ if $p=1$ and $x^{2\,p-2}\,y^{2\,q-2}$ otherwise.
\end{example}

From Propositions~\ref{prop:iter} and~\ref{prop:upperbound}, we can
deduce that $S=\acc{\frac{\fail(f)}{\LM(f)},\ f\notin I}$.%

\begin{example}\label{ex:binom_bms}
  We give the trace of the algorithm called on the binomial sequence
  $\bin$ for the
  $\DRL(y\prec x)$ ordering up to monomial $x^3$ (hence visiting
  all the monomials of degree at most $3$).

  To simplify the reading, whenever a relation succeeds in $m$ or
  cannot be tested in $m$, we
  skip the updating part as this relation remains the same.

  We start with the empty staircase $S$ and
  the relation $G=\{1\}$.
  \begin{enumerate}
  \item[] For the monomial $1$
    \begin{enumerate}
    \item[] The relation $g_1=1$ fails since $[\bin_{0,0}]=1$. Thus
      $S'=\{[1,1]\}$.
    \item[] $S'$ is updated to $\{[1,1]\}$ and $G'=\{y,x\}$.
    \item[] For the relation $g_1'=y$, $y\nmid 1$ thus $g_1'=y$.
    \item[] For the relation $g_2'=x$, $x\nmid 1$ thus $g_2'=x$.
    \item[] We update $G:=G'=\{y,x\}$ and $S:=S'=\{[1,1]\}$.
    \end{enumerate}
  \item[] For the monomial $y$
    \begin{enumerate}
    \item[] The relation $g_1=y$ succeeds since $[\bin_{0,1}]=0$.
    \item[] Nothing must be done for the relation $g_2=x$.
    \item[] $S'$ is set to $\{[1,1]\}$ and $G'=\{y,x\}$.
    \item[] We set $g_1'=y$ and $g_2'=x$.
    \item[] We update $G:=G'=\{y,x\}$ and $S:=S'=\{[1,1]\}$.
    \end{enumerate}
  \item[] For the monomial $x$
    \begin{enumerate}
    \item[] Nothing must be done for the relation $g_1=y$.
    \item[] The relation $g_2=x$ fails since $[\bin_{1,0}]=1$. Thus
      $S'=\{[1,1],[x,1]\}$.
    \item[] $S'$ is set to $\{[1,1]\}$ and $G'=\{y,x\}$.
    \item[] We set $g_1'=y$.
    \item[] For the relation $g_2'=x$, $x| x$ and
      $\frac{x}{x}|\fail(1)$, hence $g_2'=x-1$.
    \item[] We update $G:=G'=\{y,x-1\}$ and $S:=S'=\{[1,1]\}$.
    \end{enumerate}
  \item[] For the monomial $y^2$
    \begin{enumerate}
    \item[] The relation $g_1=y$ succeeds since $[\bin_{0,2}]=0$.
    \item[] Nothing must be done for the relation $g_2=x-1$.
    \item[] $S'$ is set to $\{[1,1]\}$ and $G'=\{y,x\}$.
    \item[] We set $g_1'=y$ and $g_2'=x-1$.
    \item[] We update $G:=G'=\{y,x-1\}$ and $S:=S'=\{[1,1]\}$.
    \end{enumerate}
  \item[] For the monomial $x\,y$
    \begin{enumerate}
    \item[] The relation $g_1=y$ fails since $[\bin_{1,1}]=1$. Thus
      $S'=\{[1,1],[y,x]\}$.
    \item[] The relation $g_2=x-1$ fails since $[\bin_{1,1}-\bin_{0,1}]=1$. Thus
      $S'=\{[1,1],[y,x],[x-1,y]\}$.
    \item[] $S'$ is set to $\{[y,x],[x-1,y]\}$ and $G'=\{y^2,x\,y,x^2\}$.
    \item[] For the relation $g_1'=y^2$, $y^2\nmid x\,y$ thus $g_1'=y^2$.
    \item[] For the relation $g_2'=x\,y$, $x\,y| x\,y$
      and $\frac{x\,y}{x\,y}|\fail(y)$, hence $g_2'=x\,y-1$.
    \item[] For the relation $g_3'=x^2$, $x^2\nmid x\,y$ thus $g_3'=x^2-x$.
    \item[] We update $G:=G'=\{y^2,x\,y-1,x^2-x\}$ and
      $S:=S'=\{[y,x],[x-1,y]\}$.
    \end{enumerate}
  \item[] For the monomial $x^2$
    \begin{enumerate}
    \item[] Nothing must be done for the relation $g_1=y^2$.
    \item[] Nothing must be done for the relation $g_2=x\,y-1$.
    \item[] The relation $g_3=x^2-x$ succeeds since
      $[\bin_{2,0}-\bin_{1,0}]=0$.
    \item[] $S'$ is set to $\{[y,x],[x-1,y]\}$ and $G'=\{y^2,x\,y,x^2\}$.
    \item[] We set $g_1'=y^2$, $g_2'=x\,y-1$ and $g_3'=x^2-x$.
    \item[] We update $G:=G'=\{y^2,x\,y-1,x^2-x\}$ and
      $S:=S'=\{[y,x],[x-1,y]\}$.
    \end{enumerate}
  \item[] For the monomial $y^3$
    \begin{enumerate}
    \item[] The relation $g_1=y^2$ succeeds since $[\bin_{0,3}]=0$.
    \item[] Nothing must be done for the relation $g_2=x\,y-1$.
    \item[] Nothing must be done for the relation $g_3=x^2-x$.
    \item[] $S'$ is set to $\{[y,x],[x-1,y]\}$ and $G'=\{y^2,x\,y,x^2\}$.
    \item[] We set $g_1'=y^2$, $g_2'=x\,y-1$ and $g_3=x^2-x$.
    \item[] We update $G:=G'=\{y^2,x\,y-1,x^2-x\}$ and
      $S:=S'=\{[y,x],[x-1,y]\}$.
    \end{enumerate}
  \item[] For the monomial $x\,y^2$
    \begin{enumerate}
    \item[] The relation $g_1=y^2$ succeeds since $[\bin_{1,2}]=0$.
    \item[] The relation $g_2=x\,y-1$ succeeds since $[\bin_{1,2}-\bin_{0,1}]=0$.
    \item[] Nothing must be done for the relation $g_3=x^2-x$.
    \item[] $S'$ is set to $\{[y,x],[x-1,y]\}$ and $G'=\{y^2,x\,y,x^2\}$.
    \item[] We set $g_1'=y^2$, $g_2'=x\,y-1$ and $g_3=x^2-x$.
    \item[] We update $G:=G'=\{y^2,x\,y-1,x^2-x\}$ and
      $S:=S'=\{[x,y],[y,x-1]\}$.
    \end{enumerate}
  \item[] For the monomial $x^2\,y$
    \begin{enumerate}
    \item[] Nothing must be done for the relation $g_1=y^2$.
    \item[] The relation $g_2=x\,y-1$ fails since
      $[\bin_{2,1}-\bin_{1,0}]=1$. Thus
      $S'=\{[y,x],[x-1,y],[x\,y-1,x]\}$.
    \item[] The relation $g_3=x^2-x$ fails since
      $[\bin_{2,1}-\bin_{1,1}]=1$. Thus
      $S'=\{[y,x],[x-1,y],[x\,y-1,x],[x^2-x,y]\}$.
    \item[] $S'$ is set to $\{[y,x],[x-1,y]\}$ and $G'=\{y^2,x\,y,x^2\}$.
    \item[] We set $g_1'=y^2$.
    \item[] For the relation $g_2'=x\,y$, $x\,y| x^2\,y$ and
      $\frac{x^2\,y}{x\,y}|\fail(y)$, hence $g_3'=x\,y-y-1$.
    \item[] For the relation $g_3'=x^2$, $x^2| x^2\,y$ and
      $\frac{x^2\,y}{x^2}|\fail(x-1)$, hence $g_3'=x^2-2\,x+1$.
    \item[] We update $G:=G'=\{y^2,x\,y-y-1,x^2-2\,x+1\}$ and
      $S:=S'=\{[y,x],[x-1,y]\}$.
    \end{enumerate}
  \item[] For the monomial $x^3$
    \begin{enumerate}
    \item[] Nothing must be done for the relation $g_1=y^2$.
    \item[] Nothing must be done for the relation $g_2=x\,y-y-1$.
    \item[] The relation $g_3=x^2-2\,x+1$ succeeds since
      $[\bin_{3,0}-2\,\bin_{2,0}+\bin_{1,0}]=0$.
    \item[] $S'$ is set to $\{[y,x],[x-1,y]\}$ and $G'=\{y^2,x\,y,x^2\}$.
    \item[] We set $g_1'=y^2$, $g_2'=x\,y-y-1$ and $g_3=x^2-2\,x+1$.
    \item[] We update $G:=G'=\{y^2,x\,y-y-1,x^2-2\,x+1\}$ and
      $S:=S'=\{[y,x],[x-1,y]\}$.
    \end{enumerate}
  \item[] The algorithm returns relations $y^2,x\,y-y-1,x^2-2\,x+1$,
    all three with a shift $x$.
  \end{enumerate}
\end{example}

\subsection{A Linear Algebra interpretation of the \BMS algorithm}
\label{ss:bms_lin_alg}
In order to make the presentation of the \BMS algorithm closer to that
of the \sFGLM algorithm, we propose to replace every evaluation using
the $[\,]$ operator with a matrix-vector product.

As stated above, given a monic relation $f=\LM(f)+\sum_{s\in S}\alpha_s\,s$,
testing the shift of this relation by a monomial $m$ is done with the
bracket operator, \ie testing whether $[m\,f]=0$ or not. Denoting
$\vec{f}$, the vector
\[
\vec{f}=\kbordermatrix{
  &1\\
  \vdots &\vdots\\
  s\in S &\alpha_s\\
  \vdots &\vdots\\
  \LM(f) &1
},
\] this can also
be done through testing if the following matrix-vector product
\[H_{m,S\cup\{\LM(f)\}}\,\vec{f}=
\kbordermatrix{
  &\cdots &s\in S&\cdots &\LM(f)\\
  m &\cdots &[m\,s] &\cdots &[m\,\LM(f)]
}\,
\begin{pmatrix}
  \vdots\\\alpha_s\\\vdots\\1
\end{pmatrix}
=0
\]
or not. In this setting, the definitions of the \emph{shift} and the
\emph{fail} of a relation, \ie Definition~\ref{def:shift}, become as follows.
\begin{definition}\label{def:fail_linal}
  Let $f=\LT(f)+\sum_{s\in S}\alpha_s\,s$ be a polynomial.

  The monomial $m$ is a \emph{shift of $f$} if
  \[H_{\{1,\ldots,m\},S\cup\{\LM(f)\}}\,\vec{f}=
  \kbordermatrix{
    &\cdots &s\in S&\cdots &\LM(f)\\
    1 &\cdots &[s] &\cdots &[\LM(f)]\\
    \vdots &&\vdots &&\vdots\\
    m &\cdots &[m\,s] &\cdots &[m\,\LM(f)]\\
  }\,
  \begin{pmatrix}
    \vdots\\\alpha_s\\\vdots\\1
  \end{pmatrix}
  =
  \begin{pmatrix}
    0\\\vdots\\0
  \end{pmatrix}.
  \]

  Let $m^+$ be the successor of $m$, $m^+\,\LM(f)$ is the \emph{fail of $f$} if
  \[H_{\{1,\ldots,m,m^+\},S\cup\{\LM(f)\}}\,\vec{f}=
  \kbordermatrix{
    &\cdots &s\in S&\cdots &\LM(f)\\
    1 &\cdots &[s] &\cdots &[\LM(f)]\\
    \vdots &&\vdots &&\vdots\\
    m &\cdots &[m\,s] &\cdots &[m\,\LM(f)]\\
    m^+ &\cdots &[m^+\,s] &\cdots &[m^+\,\LM(f)]
  }\,
  \begin{pmatrix}
    \vdots\\\alpha_s\\\vdots\\1
  \end{pmatrix}
  =
  \begin{pmatrix}
    0\\\vdots\\0\\e
  \end{pmatrix},
  \]
  with $e\neq 0$.
\end{definition}
We can also write another proof of Proposition~\ref{prop:augm_shift}
with a matrix viewpoint.
\begin{proof}[Proof of Proposition~\ref{prop:augm_shift}]
  Let $f_1=\LM(f_1)+\sum_{s\in S}\alpha_s\,s$ and
  $f_2=\LM(f_2)+\sum_{s\in S'}\beta_s\,s$ be monic. Let $v^-$ be the predecessor
  of $v$. Let $\tilde{S}=S\cup S'\setminus\{\LM(f_2),\LM(f_1)\}$,
  assuming $\LM(f_2)\neq\LM(f_1)$, then
  we have
  \begin{align*}
    H_{\{1,\ldots,v^-,v\},\tilde{S}\cup\{\LM(f_2),\LM(f_1)\}}\,(\vec{f_1}+c\,\vec{f_2})
    &=
      \begin{pmatrix}
        0\\\vdots\\0\\e_1+c\,e_2
      \end{pmatrix}\\
    \kbordermatrix{
    &\cdots &s\in \tilde{S}&\cdots &\LM(f_2) &\LM(f_1)\\
    1 &\cdots &[s] &\cdots &[\LM(f_2)] &[\LM(f_1)]\\
    \vdots &&\vdots &&\vdots&\vdots\\
    v^- &\cdots &[v^-\,s] &\cdots &[v^-\,\LM(f_2)] &[v^-\,\LM(f_1)]\\
    v &\cdots &[v\,s] &\cdots &[v\,\LM(f_2)] &[v\,\LM(f_1)]\\
    }\,
    \begin{pmatrix}
      \vdots\\\alpha_s+c\,\beta_s\\\vdots\\c\\1
    \end{pmatrix}
    &=
      \begin{pmatrix}
        0\\\vdots\\0\\e_1+c\,e_2
      \end{pmatrix}.
  \end{align*}
  It is now clear that vector $\vec{f}_1-\frac{e_1}{e_2}\,\vec{f}_2$
  is in the kernel of
  this matrix. That is, polynomial $f_1-\frac{e_1}{e_2}\,f_2$ has a
  shift $v$.
\end{proof}

Changing every evaluation into a matrix-vector product in the \BMS
algorithm yields 
the following presentation of the \BMS algorithm, namely
Algorithm~\ref{algo:bms_linalg}.

\begin{algorithm2e}[htbp!]\label{algo:bms_linalg}
  \small
  \DontPrintSemicolon
  \TitleOfAlgo{Linear Algebra variant of the \BMS algorithm.}
  \KwIn{A table $\bu=(u_{\bi})_{\bi\in\N^n}$ with coefficients in
    $\K$, a monomial ordering $\prec$ and a monomial $M$
    as the stopping condition.}
  \KwOut{A set $G$ of relations generating $I_M$.}
  $T := \{m\in\K[\x], m\preceq M\}$.\tcp*{ordered for $\prec$}
  $G := \{1\}$.\tcp*{the future \gb}
  $S := \emptyset$.\tcp*{staircase edge, elements will be
    $[h,\fail(h)/\LM(h)]$} 
  
  \Forall{$m\in T$}{
    $S' := S$.\;
    \For{$g\in G$}{
      \If{$\LM(g)| m$}{
        $e:=H_{\acc{\frac{m}{\LM(g)}},\supp(g)}\,\vec{g}$.\;
        \If{$e\neq 0$}{
          $S':=S'\cup
          \acc{\cro{\frac{g}{e},\frac{m}{\LM(g)}}}$.\;
        }
      }
    }
    $S':=
    \min_{\fail(h)\in S'}\acc{[h,\fail(h)/\LM(h)]}$.
    \tcp*{see Remark~\ref{rk:choose_Sm}}
    $G':= \Border (S')$.\;
    
    \For{$g'\in G'$}{
      Let $g\in G$ such that $\LM(g)|\LM(g')$.\; 
      \uIf{$\LM(g) \nmid m$}{
        $g':=\frac{\LM(g')}{\LM(g)}\,g$.\tcp*{shifts the relation}
      }
      \uElseIf{$\exists\,h\in S,
        \frac{m}{\LM(g')} |\fail(h)$}{
        $g':= \frac{\LM(g')}{\LM(g)}\,g
        -\pare{H_{\acc{\frac{m}{\LM(g)}},\supp(g)}\,\vec{g}}\,
        \frac{\LM(g')\,\fail(h)}{m}\,h$.\tcp*{see Prop.~\ref{prop:augm_shift}}
      }
      \lElse{
        $g':=g$.
      }
    }
    $G := G'$ \;
    $S := S'$ \;
  }  
  \KwRet $G$.
\end{algorithm2e}


\end{appendices}




\end{document}